\documentclass[11pt,letterpaper]{article}
\usepackage[margin=1in]{geometry}
\usepackage[T1]{fontenc}
\usepackage{lmodern,amsmath,amssymb,amsthm,booktabs,array,tabularx,microtype}
\usepackage{setspace}
\usepackage{xcolor}
\usepackage{comment}
\usepackage{tikz}
\usepackage{afterpage}
\usetikzlibrary{arrows.meta,patterns,decorations.pathreplacing,positioning,calc}

\definecolor{LFiveA}{RGB}{205,221,235}
\definecolor{LFiveB}{RGB}{222,237,224}
\definecolor{LFiveCharge}{RGB}{245,224,186}
\definecolor{LFiveGuard}{RGB}{67,73,80}
\tikzset{
  lf5/base/.style={draw=black,line width=.55pt},
  lf5/A/.style={lf5/base,fill=LFiveA},
  lf5/B/.style={lf5/base,fill=LFiveB,
    postaction={pattern=north east lines,pattern color=black!23}},
  lf5/charge/.style={lf5/base,fill=LFiveCharge,double=LFiveCharge,double distance=1pt},
  lf5/guard/.style={lf5/base,fill=LFiveGuard,text=white},
  lf5/idle/.style={draw=black!55,densely dashed,fill=white,line width=.45pt},
  lf5/axis/.style={-{Stealth[length=1.5mm]},line width=.45pt},
  lf5/guide/.style={draw=black!32,densely dotted,line width=.4pt},
  lf5/brace/.style={decorate,decoration={brace,amplitude=3pt},line width=.4pt},
  lf5/release/.style={isosceles triangle,shape border rotate=90,
    inner sep=0pt,minimum width=3.7pt,minimum height=3.7pt,
    draw=black,fill=black},
  lf5/due/.style={diamond,inner sep=0pt,minimum size=4.8pt,
    draw=black,fill=white,line width=.6pt}
}
\usetikzlibrary{shapes.geometric}

\newcommand{\LFiveJob}[4]{%
  \draw[lf5/#1] (#2,.35) rectangle (#3,1.05);
  \node[font=\small,text=black] at ({(#2+#3)/2},.70) {#4};}
\newcommand{\LFiveGuardJob}[3]{%
  \draw[lf5/guard] (#1,.35) rectangle (#2,1.05);
  \node[text=white,font=\small] at ({(#1+#2)/2},.70) {#3};}
\newcommand{\LFiveRelease}[1]{\node[lf5/release] at (#1,.28) {};}
\newcommand{\LFiveDue}[1]{\node[lf5/due] at (#1,1.05) {};}
\newcommand{\LFiveTick}[3]{%
  \draw (#1,.05)--(#1,-.05);
  \node[anchor=north,inner sep=1pt] at (#1,#2) {#3};}

\newif\ifTWTshowblue
\TWTshowbluetrue

\usepackage[round,authoryear]{natbib}
\usepackage{xurl}
\usepackage[hidelinks]{hyperref}
\hypersetup{pdfauthor={},pdftitle={Strong NP-Hardness and Approximation Algorithm for Equal-Processing-Time Weighted Tardiness}}
\newtheorem{theorem}{Theorem}
\newtheorem{lemma}{Lemma}
\newtheorem{proposition}{Proposition}

\newcommand{\OPT}{\operatorname{OPT}}
\newcommand{\UB}{\overline K}
\newcommand{\FS}{F_{\mathrm{sum}}}
\newcommand{\problem}{1\mid r_j,p_j=p\mid\sum_jw_jT_j}
\newcolumntype{Y}{>{\raggedright\arraybackslash}X}
\title{Strong NP-Hardness and Approximation Algorithm for Weighted Tardiness with Release Dates and Identical Processing Times}
\author{Zhi-Long Chen\\[6pt]{\normalsize University of Maryland}
\and Nicholas G. Hall\\[6pt]{\normalsize The Ohio State University}}
\date{September 23, 2026}
\begin{document}
\begin{titlepage}
\hypersetup{pageanchor=false}
\thispagestyle{empty}
\centering
\makeatletter
\vspace*{1.05in}
\begin{minipage}{0.94\textwidth}
\centering
\setstretch{1.15}
{\fontsize{20}{26}\selectfont\bfseries \@title\par}
\end{minipage}
\par\vspace{1.45in}
{\fontsize{14}{18}\selectfont
\def\and{\end{tabular}\hspace{0.60in}\begin{tabular}[t]{c}}%
\begin{tabular}[t]{c}\@author\end{tabular}\par}
\vspace{3.45in}
{\large \@date\par}
\makeatother
\vfill
\end{titlepage}

\begin{titlepage}
\thispagestyle{empty}
\begin{abstract}
We study nonpreemptive scheduling on a single machine with release dates, due dates, positive job weights, and a common processing time. The objective is to minimize total weighted tardiness. Although closely related equal-processing-time problems admit polynomial-time algorithms, the complexity of this problem has remained open in the literature since 2010. We prove that its decision version is strongly NP-complete, even when every job can meet its due date if processed immediately upon release. The reduction is from unweighted MAX-CUT and uses a quadratic number of jobs with polynomially bounded numerical data. Its main ingredient is a constructive normalization theorem that converts every sufficiently inexpensive feasible schedule into a binary choice for each graph vertex; after normalization, total weighted tardiness equals a constant minus a scaled cut value. We also give a deterministic polynomial-time phase-grid assignment algorithm for the shifted objective $\Phi=F+p\sum_jw_j$, where $F$ is total weighted tardiness. The algorithm enumerates at most $N$ release-date residues modulo $p$, solves one minimum-cost assignment problem for each residue, and returns the best phase-grid schedule. It runs in $O(N^5)$ arithmetic operations and achieves the tight ratio $3/2-1/(2N)$ for this algorithm. Because the added term $p\sum_jw_j$ is independent of how the jobs are scheduled, the shifted and original objectives have exactly the same optimal schedules. However, the approximation guarantee applies to the shifted objective; for the original objective, the analysis provides an additive bound. Thus, the paper both resolves the long-standing complexity question and provides a complementary worst-case guarantee for the phase-grid assignment algorithm.
\end{abstract}
\vspace{2em}
\noindent\textbf{Subject classifications:} Analysis of algorithms: computational complexity; production/scheduling: deterministic single machine.\\
\textbf{AMS classifications (2020):} Primary 90B35; Secondary 68Q25, 90C59.\\
%{\color{blue}\textbf{Area of review:} Optimization.}\\
\par\vspace{1em}

\end{titlepage}
\hypersetup{pageanchor=true}
\setcounter{page}{1}

\section{Introduction}\label{sec:intro}
We consider the problem of scheduling $N$  jobs  for processing on a single continuously available machine to minimize the total weighted tardiness where jobs have individual release dates and identical processing times. Using the three-field scheduling notation of \citet{graham1979},  where the first field describes the machine environment, the second specifies job restrictions, and the third gives the objective, this problem is denoted as $1\mid r_j,p_j=p\mid\sum_jw_jT_j$, where each job $j$ is associated with a release date $r_j$, an identical processing time $p_j=p$, a due date $d_j$, and a priority weight $w_j$. In a given schedule, the completion time of job $j$ is denoted as $C_j$ and its tardiness is defined as $T_j=\max\{C_j-d_j,0\}$. %In this problem, all the processing times are identical, i.e., $p_j=p$. 
Without loss of generality, it is assumed that $p$ and $w_j$'s are positive integers and $r_j$'s and $d_j$'s are nonnegative integers.
%{\color{red} The input consists of a positive integer common processing time $p$, positive integer weights $w_j$, and nonnegative integer release and due dates $r_j,d_j$, all encoded in binary. The number of jobs satisfies $N\ge1$.} 

Scheduling to meet due dates is a basic operations research problem. When work becomes available at different times, a decision maker must determine both which job to process next and whether to wait for a job that has not yet arrived. Total weighted tardiness measures the resulting delay by charging each job a penalty proportional to the time by which it misses its due date. The weights allow the penalty rates to differ across jobs.
 Equal processing times describe a natural restriction in which every job requires the same amount of machine capacity, while availability, urgency, and delay costs remain heterogeneous.
This restriction often changes the computational character of a scheduling problem. It removes the possibility of encoding an arbitrary collection of item sizes directly in processing times. It also permits a polynomially bounded set of candidate start times. %{\color{green} By \eqref{eq:leftshift}, every earliest start belongs to $\{r_j+kp:1\le j\le N,\ 0\le k<N\}$.} 
For several due-date objectives, these properties support dynamic programming or assignment algorithms. In particular, there are known polynomial-time algorithms for the total unweighted tardiness problem and the weighted number of tardy jobs problem  with release dates and equal processing times, i.e., $1\mid r_j,p_j=p\mid\sum_j T_j$ and $1\mid r_j,p_j=p\mid\sum_j w_jU_j$  \citep{baptiste1999,baptiste2000}, where $U_j=\mathbf 1_{\{C_j>d_j\}}$ is the tardy-job indicator. It is therefore natural to ask whether the  total weighted tardiness problem, i.e.,  $1\mid r_j,p_j=p\mid\sum_jw_jT_j$, is also tractable.

The answer has remained unresolved in the literature devoted to this problem. \citet[pp.~561 and~575]{akker2010} explicitly identify its computational complexity as open. \citet[p.~853]{gafarov2020} subsequently describe it as a minimal open problem in the scheduling classification and conjecture NP-hardness in their complexity discussion and concluding remarks. %These statements concern exactly the nonpreemptive single-machine model studied here, with arbitrary release dates and a common processing time. They provide the historical basis for the question addressed in this paper.

We prove that the decision problem of $1\mid r_j,p_j=p\mid\sum_jw_jT_j$ is strongly NP-complete  (Theorem~\ref{thm:main}). The result holds even when every due date is at least one processing time after the corresponding release, i.e., $d_j\ge r_j + p$ for every job. 
%so infeasibility of an individual job's due date is not the source of hardness. 
Every numerical parameter in the reduction, including the objective threshold, is polynomially bounded in the number of jobs. 
%{\color{green} Table~\ref{tab:magnitudes} and the polynomial-magnitude part of the proof of Theorem~\ref{thm:main} establish these bounds.} 
This resolves the published complexity open question and strengthens the NP-hardness conjecture of \citet{gafarov2020}.

The proof encodes unweighted MAX-CUT. Each vertex is represented by a block of equal-length jobs with two possible normalized behaviors. A pair of jobs at each incidence determines which weight remains for processing in a final portion of the schedule. Additional equal-length jobs enforce the timing of the blocks through the objective itself. The central argument shows constructively that any schedule meeting the reduction's cost bound can be converted, without increasing cost, into one of these normalized schedules  (Proposition~\ref{prop:normalization}). %Thus the proof accounts for schedules that interleave jobs unexpectedly, leave idle time, or initially use noninteger start times.

The normalized objective has an exact pairwise decomposition. Its coefficients can be chosen independently, which permits a graph interaction to be represented through job weights and due dates (Proposition~\ref{prop:coefficients}, especially \eqref{eq:cutcost}). This provides the needed coupling while all processing times remain identical. %The proof uses neither externally imposed machine interruptions nor precedence constraints. 
%Classical assignment, time-discretization, and weighted-order arguments retain their usual roles; 

The hardness result motivates the search for algorithms with provable
performance. A direct multiplicative guarantee for total weighted tardiness is
delicate because its optimum may be zero. We therefore analyze the shifted objective
\[
 \Phi(S)=F(S)+p\sum_jw_j,
\]
where the added term $p\sum_jw_j$ is independent of how the jobs are
scheduled. Consequently, $\Phi$ and the original objective $F$ have the same
ordering of schedules and the same minimizers. Our deterministic phase-grid assignment algorithm enumerates every
distinct release-date residue modulo $p$. For each residue, it rounds releases
upward to the associated grid, solves the resulting minimum-cost assignment
problem exactly, and retains the best schedule. The algorithm uses at most
$N$ assignments, runs in $O(N^5)$ arithmetic operations, and satisfies
\[
 \frac{\Phi_{\mathrm{alg}}}{\Phi^*}
 \le \frac32-\frac{\sum_jw_j^2}{2(\sum_jw_j)^2}
 \le \frac32-\frac1{2N}.
\]
The last bound is tight for the specified phase-grid assignment algorithm. Equivalently, the
analysis gives an explicit additive guarantee for $F$; it does not claim a
constant multiplicative approximation for the unshifted objective. The
contributions are therefore the construction and normalization that encode a
cut problem, together with the phase-enumeration analysis that converts
related grid and assignment ideas into a worst-case guarantee for a precisely
defined shifted objective.

Section~\ref{sec:literature} reviews the complexity and algorithmic results most directly surrounding the problem, including grid-based methods and shifted-objective analyses. %Section~\ref{sec:model} defines the model and records the time-domain facts needed for the reduction. 
Sections~\ref{sec:construction} and~\ref{sec:normalization} give the construction and normalization argument. Section~\ref{sec:identity} derives the cost identity and proves strong NP-completeness. Section~\ref{sec:approximation} presents the phase-grid assignment algorithm and its shifted-objective guarantee, and Section~\ref{sec:conclusion} concludes.

\section{Related literature}\label{sec:literature}
%We use the three-field scheduling notation of \citet{graham1979}. The first field describes the machine environment, the second specifies job restrictions, and the third gives the objective. We write $T_j=\max\{C_j-d_j,0\}$ for tardiness and $U_j=\mathbf 1_{\{C_j>d_j\}}$ for the tardy-job indicator. Equal durations are denoted by $p_j=p$. 
%The machine symbols $P_m$ and $P$ denote a fixed and an input number of identical parallel machines, respectively; they are unrelated to the lowercase common duration $p$. 
Table~\ref{tab:literature} summarizes the complexity results of the single-machine and parallel-machine problems related to the studied problem. We briefly discuss some of the problems and results given in the table and other related literature in the following subsections. 

\begin{table}[htbp]
\caption{Complexity results surrounding equal-processing-time weighted tardiness}\label{tab:literature}
\scriptsize 
\setstretch{1.1}
\begin{tabularx}{\textwidth}{@{}>{\raggedright\arraybackslash}p{0.35\textwidth}>{\raggedright\arraybackslash}p{0.32\textwidth}Y@{}}
\toprule
Problem  & Result & Source\\
\midrule
$1\mid\mid\sum_jT_j$ & Ordinary NP-hardness; pseudo-polynomial algorithm and FPTAS & \citet{lawler1977,lawler1982,duleung1990}\\[1ex]
$1\mid\mid\sum_jw_jT_j$ & Strongly NP-hard & \citet{lenstra1977}\\[1ex]
$1\mid r_j,p_j=p\mid\sum_jT_j$ & Polynomial, $O(N^7)$ algorithm & \citet{baptiste2000}\\[1ex]
$1\mid r_j,p_j=p\mid\sum_jw_jU_j$ & Polynomial, $O(N^7)$ algorithm & \citet{baptiste1999}\\[1ex]
$1\mid r_j,p_j=p\mid\sum_jU_j$ & Polynomial, $O(N^5)$ algorithm & \citet{chrobak2006}\\[1ex]
$1\mid r_j=a+k_jp,p_j=p\mid\sum_j w_jT_j$ & Polynomial, reduction to assignment   & \citet{gafarov2020} % Proposition~\ref{prop:grid}
\\[1ex]
$1\mid r_j,p_j=p\mid\sum_jw_jT_j$ & Explicitly left open; NP-hardness conjectured & \citet{akker2010,gafarov2020}\\[1ex]
{\bf Same problem, including $d_j\ge r_j+p$} & {\bf Strongly NP-hard} & {\bf Theorem~\ref{thm:main}} \\ 
$P_m\mid r_j,p_j=p\mid\sum_jT_j$, fixed $m$ & Polynomial, $O(N^{3m+4})$ algorithm & \citet{baptiste2000}\\[1ex]
$P \mid r_j,p_j=p\mid\sum_jT_j$, input machine count & Polynomial, reduction to LP  & \citet{brucker2005scheduling}\\[1ex]
$P_m\mid r_j,p_j=p\mid\sum_jw_jU_j$, fixed $m$ & Polynomial, $O(N^{6m+1})$ algorithm & \citet{bbkt2004}\\[1ex]
$P\mid r_j,p_j=p\mid\sum_jU_j$, input machine count & NP-hard & \citet{heegermolter2025}\\
\bottomrule
\end{tabularx}
\vskip 0.15cm
\parbox{\textwidth}{\footnotesize Note: $N$ is the number of jobs. %All rows are nonpreemptive. 
%In the first two rows, processing times are arbitrary and all jobs are initially available. ``Polynomial'' does not assert a strongly polynomial bound unless stated in the text. 
The symbol $P$ in the machine field denotes an arbitrary number of identical machines; $P_m$ denotes a given number ($m$) of identical machines; %$p_j=p$ denotes common job duration. 
%The open-status row records the cited historical statements. Additional reported LP-based special cases and the later qualification are discussed in Section~\ref{sec:literature}.
}
\end{table}

\subsection{Complexity results}
For arbitrary processing times and no release dates, the distinction between weighted and unweighted tardiness is well established. \citet{lawler1977} gives a pseudo-polynomial dynamic program for $1\mid\mid\sum_jT_j$, and \citet{duleung1990} prove NP-hardness in the ordinary sense. \citet{lawler1982} also gives a fully polynomial approximation scheme. %Together these results show why a proof of NP-hardness alone leaves an important algorithmic question open for a numerical scheduling problem.

The unrestricted weighted counterpart $1\mid\mid\sum_jw_jT_j$ is strongly NP-hard \citep{lenstra1977}. That result does not establish hardness under a common processing time: a hardness proof for a larger class does not automatically survive the equal-duration restriction. Indeed, when the processing times are identical, i.e., $p_j=p$, if all jobs are available at time zero, the completion positions can be fixed and the jobs assigned to  positions at costs $w_j\max\{kp-d_j,0\}$. Thus, $1\mid p_j=p \mid\sum_jw_jT_j$ reduces to an assignment problem and can be solved in polynomial time \citep{kuhn1955}. This is discussed in \citet{gafarov2020}. They also point out that a more general common-grid version of the problem, where $r_j = a + k_jp$, for a common nonnegative integer $a$ and a job dependent integer multiplier $k_j$, can  be modeled as an assignment problem.

\citet{baptiste1999} establishes strongly polynomial dynamic programs for minimizing the weighted number of tardy equal-length jobs with release dates, both without and with preemption. His nonpreemptive algorithm has running time $O(N^7)$ for $N$ jobs. \citet{baptiste2000} obtains polynomial algorithms for equal-length scheduling on a fixed number of identical parallel machines, including the total-tardiness objective. The single-machine result gives a polynomial algorithm for $1\mid r_j,p_j=p\mid\sum_jT_j$. Thus replacing tardiness amounts by tardy indicators, or replacing arbitrary weights by equal weights, yields tractable neighbors of our problem.

Subsequent work develops this algorithmic boundary in several directions. \citet{bbkt2004} give polynomial algorithms for ten equal-processing-time problems, one of which is the weighted number of tardy jobs problem on a fixed number of identical machines. \citet{carlier1981} gives a polynomial algorithm for deciding whether all equal-length jobs can meet their deadlines and proposes an extension to maximize throughput. \citet{chrobak2006} identify an error in that maximization extension and provide an $O(N^5)$ algorithm for the unweighted single-machine throughput problem. Their counterexample concerns the maximization extension; the feasibility result remains valid.

Throughput and tardiness use the same release-date and due-date information but value late jobs differently. In a throughput problem, the cost associated with a job declared late is independent of how late it is. The remaining jobs can therefore be appended after the chosen on-time subset without changing the selection objective. Under total weighted tardiness, their completion times still contribute to the objective. Our reduction exploits exactly this dependence: the jobs left for the final part of the schedule carry the interaction among the earlier binary choices.

For the weighted-tardiness problem studied here, \citet{akker2010} explicitly
leave the general computational complexity open, and \citet{gafarov2020}
subsequently conjecture NP-hardness and reiterate the open status. Theorem~\ref{thm:main}
settles this question by proving strong NP-hardness, even under
$d_j\ge r_j+p$. The result also resolves a related classification question in
multiple-due-date scheduling. In the equal-processing-time model MDS-EP of
\citet{kuehn2024}, every job has several due dates, each carrying a penalty if
missed. Their Corollary~2 gives a polynomial reduction from $\problem$ to
MDS-EP. Composing that reduction with Theorem~\ref{thm:main} proves that
MDS-EP is NP-hard; this observation does not assert preservation of strong
NP-hardness or approximation guarantees.

Parallel-machine classifications depend on the number of machines and on
preemption assumptions \citep{kravchenkowerner2011}. Equal-length total
tardiness is polynomial for a fixed number of identical machines
\citep{baptiste2000}, and equal-length weighted throughput is polynomial for a
fixed number of nonpreemptive machines \citep{bbkt2004}. When the number of
machines is part of the input, hardness is known for the preemptive
weighted-count problem \citep{bruckerkravchenko2006} and for the
nonpreemptive unweighted-count problem with release dates
\citep{heegermolter2025}; the latter resolves the question discussed by
\citet{sgall2012}. These parallel-machine results concern tardy-job counts or
unweighted tardiness under different machine assumptions, rather than the
single-machine weighted-tardiness objective considered here.

\subsection{Exact, heuristic, and grid-based algorithms}
Exact and heuristic procedures for the more general problem
$1\mid r_j\mid\sum_jw_jT_j$ also apply when processing times are equal.
\citet{akturk2000} develop a branch-and-bound method based on time-dependent
dominance rules and lower bounds, and \citet{akturk2001} use related dominance
information in constructive and local-search heuristics. These methods do not
exploit the common processing time through residue grids or assignment.
For $\problem$, \citet{akker2010} develop a time-indexed formulation over the
union of the release-generated grids, together with LP-based structural
arguments, branching, and column generation. They report tractable special
cases with common due dates, common weights, common releases, or equally
ordered release and due dates. Their formulation retains capacity interactions
among grids with different residues, whereas our algorithm separates the
residue phases and solves one assignment problem for each phase. Their
computational procedures address exact solution of instances, but favorable
behavior of the LP relaxation does not provide a polynomial worst-case bound
for the general problem.

For the equal-processing-time problem itself, \citet{gafarov2020} give two
exact algorithms with running time $O(N^3 2^N)$. One enumerates subsets of jobs
fixed to start at their release dates and solves an assignment problem for the
remaining jobs. This has the same high-level pattern as our method---enumerate
a structural choice, solve an assignment problem, and select the best
solution---but its enumeration is exponential and exact, whereas our
algorithm examines at most $N$ residue phases and is polynomial and
approximate. They also study dominance properties, special cases, and exact
search, and question the completeness of an LP-conversion argument in the
earlier work. Our hardness proof and approximation analysis do not rely on
that conversion argument.

Time grids and assignment reductions also occur in neighboring due-date
problems. \citet{dessouky1990} reduce identical-job scheduling without
individual release restrictions to linear assignment, obtaining an $O(N^3)$
algorithm for general nondecreasing costs and $O(N\log N)$ algorithms for
total tardiness and the weighted number of tardy jobs. \citet{zhang2025} use
generated completion times for bicriterion equal-length scheduling involving
total tardiness and the number of tardy jobs on uniform parallel machines,
with an $O(N^2)$ Pareto algorithm and $O(N\log N)$ algorithms for two
hierarchical variants. The polynomial dynamic programs summarized in the
preceding subsection likewise restrict attention to times of the form $r_i+kp$, but
optimize jointly over the union of all release-generated grids. In contrast,
our algorithm solves an assignment problem independently for each residue
class and then compares the phase optima. \citet{gafarov2020} show that
selecting one arbitrary grid can have
unbounded relative error for the original objective; Section~\ref{sec:approximation}
obtains a bounded additive loss by testing every release phase and evaluates
that loss under a shifted objective.

\subsection{Shifted-objective performance guarantees}
There is direct precedent for using a translated objective when tardiness can
vanish. For $P_m\mid d_j=d\mid\sum_jT_j$, \citet{kovalyov2002} construct
approximation schemes satisfying
\[
 \frac{T^A-T^*}{T^*+d}\le\varepsilon,
 \qquad\text{equivalently}\qquad
 \frac{T^A+d}{T^*+d}\le1+\varepsilon.
\]
Thus the common due date is added explicitly to the performance measure. For
fixed $m$, their two algorithm families run in
$O(N^{2m}/\varepsilon^{m-1})$ and $O(N^{m+1}/\varepsilon^m)$ time; for
$m=2$, the first improves to $O(N^3/\varepsilon)$.

\citet{kolliopoulos2006} study the broader family
$\alpha\mid\beta\mid\sum_jw_j(T_j+d_j)$ on single, identical-parallel,
uniformly related, and unrelated machines, with variants involving release
dates, precedence constraints, and preemption. Their models include, for
example, $P\mid r_j\mid\sum_jw_j(T_j+d_j)$,
$1\mid r_j,\mathrm{prec}\mid\sum_jw_j(T_j+d_j)$, and
$R\mid r_j\mid\sum_jw_j(T_j+d_j)$. They transfer approximation results from
weighted completion time and give an $O(N^3/\varepsilon)$ FPTAS for
$1\mid d_j=D\mid\sum_jw_j(T_j+d_j)$. Since
$\sum_jw_j(T_j+d_j)=\sum_jw_jT_j+\sum_jw_jd_j$, the modification preserves
the optimal schedules while changing the approximation ratio.

A related transformation is standard for maximum lateness. If
$D\ge\max_jd_j$ and $q_j=D-d_j$, then
$\max_j\{C_j+q_j\}=D+L_{\max}$. In the release-date/delivery-time
formulation, \citet{potts1980} give a $3/2$-approximation for
$1\mid r_j\mid L_{\max}$, and \citet{mastrolilli2003} develop PTASs for
$P\mid r_j\mid L_{\max}$ and its single-machine special case. These two results apply to the delivery completion objective $Q_{\max}$, which is equivalent to $D + L_{\max}$ under the displayed transformation. These studies
support the use of a schedule-independent shift as a performance measure, but
the guarantee must be identified as applying to the translated objective, not
as an ordinary multiplicative approximation for unshifted tardiness.

\subsection{Position of our result}
Our complexity result establishes strong NP-hardness while retaining one
machine, equal job lengths, arbitrary release dates, and the usual
weighted-tardiness objective. It strengthens the published NP-hardness
conjecture and separates this model from both the unweighted total-tardiness
problem and the weighted-throughput problem. Complementing that negative
result, Section~\ref{sec:approximation} combines residue-class enumeration and
assignment with a shifted-objective analysis. We found no previous algorithm
with this combination for $\problem$. The complexity conclusion does not
contradict the usefulness of existing exact methods on structured or
moderate-size instances, and the shifted guarantee does not imply a
multiplicative approximation for the original objective.

Several ingredients remain standard and are credited where used: the scheduling notation, candidate-time dominance, bipartite assignment, the weighted completion-time interchange rule of \citet{smith1956}, and the NP-completeness of unweighted MAX-CUT established by \citet{gjs1976}. The new contributions are the equal-length scheduling construction, its normalization for arbitrary bounded-cost witnesses, its exact encoding of independently specified positive edge coefficients, and the phase-averaging analysis of the phase-grid assignment algorithm.

\section{Reduction from unweighted MAX-CUT}\label{sec:construction}
% A8.4.1 / A9.5: conceptual construction roadmap
 In a feasible schedule, each job $j$, for $j=1, \ldots, N$, is executed in $[s_j,s_j+p)$, where $s_j\ge r_j$. The intervals are pairwise disjoint. Preemption is forbidden and idle time is permitted. The objective is
\[
\sum_{j=1}^{N} w_j\max\{0,C_j-d_j\},\qquad \mbox{where} ~C_j=s_j+p.
\]
The decision problem asks whether this cost is at most a given nonnegative integer $K$. All numerical input is binary encoded. Start times may be real. The common processing time is part of the instance, and can vary between instances.

The main strong NP-completeness result is stated and proved in Section~\ref{sec:identity}.

We first record the left-shifting argument underlying candidate-time discretization in equal-length scheduling \citep{bbkt2004,akker2010}. It also specifies why allowing real start times does not invalidate the later integral arguments.
For any feasible schedule, let $\pi$ be its job permutation. Define its earliest schedule by
\begin{equation}
\widehat C_0=0,\qquad
\widehat s_{\pi(k)}=\max\{r_{\pi(k)},\widehat C_{k-1}\},\qquad
\widehat C_k=\widehat s_{\pi(k)}+p.
\label{eq:leftshift}
\end{equation}
Here $k$ is the position in $\pi$, and $\widehat C_k$ is the completion time of job $\pi(k)$ in the earliest schedule.
For $k=1$, feasibility gives $s_{\pi(1)}\ge r_{\pi(1)}$, so \eqref{eq:leftshift} gives $\widehat s_{\pi(1)}\le s_{\pi(1)}$. For $k>1$, assume $\widehat C_{k-1}\le C_{\pi(k-1)}$. Feasibility gives $s_{\pi(k)}\ge r_{\pi(k)}$ and $s_{\pi(k)}\ge C_{\pi(k-1)}$; hence \eqref{eq:leftshift} gives $\widehat s_{\pi(k)}\le s_{\pi(k)}$. Adding $p$ proves $\widehat C_k\le C_{\pi(k)}$ for every $k$. The new schedule is feasible, has integer starts, and has no greater cost, because tardiness is nondecreasing in completion time. This argument applies to every feasible schedule, including one with noninteger starts, without any optimality assumption. In particular, the minimum over the finitely many earliest permutations attains the optimum. All uses of a one-unit lower bound on positive tardiness below explicitly require an integral schedule, obtained by \eqref{eq:leftshift} when necessary.

The source problem for the reduction is unweighted MAX-CUT \citep{gjs1976}. The reduction encodes each vertex by a block of equal-length jobs. For every vertex--edge incidence, the construction creates two alternative jobs. In the schedules that survive the normalization argument, exactly one of these two jobs is processed in the corresponding vertex block and the other is postponed to a final tail. This binary choice represents the side of the cut to which the vertex is assigned. The remaining jobs and parameters are chosen so that the block structure is enforced, while the weighted tardiness contributed by the tail decreases exactly when the endpoints of an edge are assigned to different sides of the cut.

Section~\ref{sec:construction-data} gives the full construction of the scheduling instance, including the auxiliary complete graph, the incidence indexing, the job data, the block boundaries, and the decision threshold. Section~\ref{sec:construction-bounds} records the parameter inequalities that are used later in the schedule argument. These inequalities ensure that the intended block structure is forced for every schedule whose cost is below the reference bound. Section~\ref{sec:normalization} then shows that such schedules can be transformed into a canonical form, and Section~\ref{sec:identity} evaluates the cost of canonical schedules to obtain the correspondence with MAX-CUT.

\begin{comment}
{\color{red}
Table~\ref{tab:roles} previews the roles of the quantities defined below.}

% A8.4.3: compact role-based notation table
\begin{table}[htbp]
{\color{red}
\centering
\caption{Roles of the construction quantities, defined in Section~\ref{sec:construction-data}; vertex modes are formalized in Section~\ref{sec:identity}}\label{tab:roles}
\small
\setstretch{1.1}
\begin{tabularx}{\textwidth}{@{}>{\raggedright\arraybackslash}p{0.30\textwidth}Y@{}}
\toprule
Quantities & Roles\\
\midrule
$n,m,M,E,N$ & Retained vertices; incidences per vertex; all incidences; auxiliary edges; scheduling jobs.\\[0.5ex]
$t_e,k_{v,e},g_{v,e}$ & Edge order; local incidence rank; global incidence rank.\\[0.5ex]
$\Delta_e,[\ell_e,h_e]$ & Interaction coefficient and its disjoint weight interval.\\[0.5ex]
$U,P,L,F_v$ & Identity-enforcing weight baseline; common duration; incidence due-date spacing; mode activation charge.\\[0.5ex]
$P_0,L_0,f_v,G,Q$ & Unscaled duration, spacing, charge coefficients, and reference cost; common integer scaling factor.\\[0.5ex]
$b_v,H$ & Vertex-block boundaries and the start of the tail.\\[0.5ex]
$\UB,K_*,x_v$ & Reference cost; decision threshold; binary vertex mode.\\
\bottomrule
\end{tabularx}
\par}
\end{table}
\end{comment}

\subsection{Source instance and job data}\label{sec:construction-data}
\label{par:preprocessing}

\begin{comment}
Take an unweighted simple graph $\mathcal G$ with edge set $\mathcal F$ and an integer $\kappa$. MAX-CUT asks whether some bipartition cuts at least $\kappa$ edges. {\color{green} If $\kappa\le0$, output a fixed YES scheduling instance: one job with $P=1$, $r=0$, $d=1$, $w=1$, and threshold zero. If $\kappa>|\mathcal F|$, output a fixed NO instance: two such jobs and threshold zero; its optimum is one. These cases also handle graphs with fewer than two vertices. In the remaining cases, delete isolated vertices and relabel the remaining vertices $1,\ldots,n$. This preserves every achievable cut value and ensures that the number $n$ used below is at most $2|\mathcal F|$, even if the original graph was specified by an edge list and a binary vertex count.} Hence assume
\[
2\le n\le2|\mathcal F|,\qquad 1\le\kappa\le|\mathcal F|.
\]

\paragraph{Original and added edges.}
Write $\mathcal G=(\mathcal V,\mathcal F)$, where $\mathcal V=\{1,\ldots,n\}$ is the retained vertex set. The construction uses the auxiliary complete graph $K_n=(\mathcal V,\mathcal E)$, with 
\end{comment}
Take an unweighted simple graph $\mathcal G=(\mathcal V,\mathcal F)$, where $\mathcal V=\{1,\ldots,n\}$ is the vertex set $\mathcal F$ the edge set, and an integer $\kappa$. MAX-CUT asks whether some bipartition cuts at least $\kappa$ edges.  To avoid trivial instances of the problem, we assume that there are no isolated vertices in $\mathcal G$, and $1\le \kappa\le |\mathcal F|$. % output a fixed YES scheduling instance: one job with $P=1$, $r=0$, $d=1$, $w=1$, and threshold zero. If $\kappa>|\mathcal F|$, output a fixed NO instance: two such jobs and threshold zero; its optimum is one. These cases also handle graphs with fewer than two vertices. 
 Hence, 
$2\le n\le2|\mathcal F|$. 

\paragraph{Auxiliary complete graph.}
 The construction uses the auxiliary complete graph $K_n=(\mathcal V,\mathcal E)$, with
\[
\mathcal E=\{\{a,b\}:1\le a<b\le n\}.
\]
Edges in $\mathcal F$ are called \emph{original edges}; edges in $\mathcal E\setminus\mathcal F$ are called \emph{added edges}. No vertices are added. The MAX-CUT objective counts only original edges, whereas every edge of the auxiliary complete graph participates in the scheduling construction. Define
\begin{equation}
\Gamma=\left\lfloor\frac{n^2}{4}\right\rfloor+1,\qquad
\Delta_e=1+\Gamma\,\mathbf{1}_{\{e\in\mathcal F\}}
\quad(e\in\mathcal E).
\label{eq:graphcoefficients}
\end{equation}
Thus $\Delta_e=\Gamma+1$ for an original edge and $\Delta_e=1$ for an added edge. Put $m=n-1$, $M=nm$, and $E=|\mathcal E|=M/2=\binom n2$. Each vertex will be associated with a \emph{vertex block}, defined below as a time interval.

\paragraph{Edge indices and weight intervals.}
% A8.4.3: concise edge-order convention
Write each unordered edge as $e=\{a,b\}$ with $a<b$, order $\mathcal E$ lexicographically by $t_e\in\{1,\ldots,E\}$, and define

\begin{equation}
\ell_e=t_e+\sum_{\substack{f\in\mathcal E\\t_f<t_e}}\Delta_f,\qquad
h_e=\ell_e+\Delta_e
\quad(e\in\mathcal E).
\label{eq:intervals}
\end{equation}
The sum is over all edges preceding $e$ in this order, including original and added edges. The intervals $[\ell_e,h_e]$ are disjoint, with a one-unit gap between successive intervals.  Their maximum endpoint is
\begin{equation}
V=E+\sum_{e\in\mathcal E}\Delta_e.
\label{eq:V}
\end{equation}
The scalar $V$ is distinct from the vertex set $\mathcal V$.

\paragraph{Vertex--edge incidences and their ranks.}
% A8.4.3: retain one precise incidence definition
For each vertex $v\in\mathcal V$, let
\[
\mathcal E_v=\{e\in\mathcal E:v\in e\}
=\{\{a,v\}:1\le a<v\}\cup\{\{v,b\}:v<b\le n\}.
\]
The set $\mathcal E_v$ contains all $m=n-1$ edges incident to $v$, with both smaller and larger other endpoints. A \emph{vertex--edge incidence} $(v,e)$ records a vertex $v$ and an edge $e\in\mathcal E_v$ containing it. An edge $e=\{a,b\}=\{b,a\}$ has two incidences, $(a,e)$ and $(b,e)$. There are $M=2E$ incidences in total. Each incidence will receive its own job pair, hence one edge gives two such job pairs, one at each endpoint.

At a fixed vertex $v$, order its incident edges by increasing other endpoint. For $e=\{a,b\}$ with $a<b$, the local incidence rank and global incidence rank are defined, respectively, 
\[
k_{v,e}=\begin{cases}b-1,&v=a,\\a,&v=b,\end{cases}
\qquad \mbox{and} ~~g_{v,e}=(v-1)m+k_{v,e}.
\]
Local ranks run from $1$ to $m$. The absence of a loop $\{v,v\}$ explains the two cases in the formula: for $v=a$, the other endpoint $b$ has rank $b-1$, whereas for $v=b$, the other endpoint $a$ has rank $a$. Equivalently, the local list is ordered by increasing other endpoint. Global ranks run from $1$ to $M$, with these local lists concatenated in increasing vertex order. Thus the vertex is ordered first and, within that vertex, the incident edges are ordered by increasing other endpoint (equivalently, by the lexicographic order of their canonical labels restricted to $\mathcal E_v$). %The edge index $t_e$ identifies one of $E$ edges; the global rank $g_{v,e}$ identifies one of $M$ incidences. 
Write $e(g)$ for the unique edge associated with incidence rank $g$. For each edge $e$, its two incidences share the same quantities $t_e,\Delta_e,\ell_e,h_e$.

% A8.4.3: replace the general pattern by a small example
Table~\ref{tab:ranks} displays the rank pattern for general $n$.

\begin{table}[htbp]
\centering
\caption{Local and global incidence ranks for general $n$}\label{tab:ranks}
\small
\setstretch{1.1}
\begin{tabularx}{\textwidth}{@{}cY>{\raggedright\arraybackslash}p{0.17\textwidth}>{\raggedright\arraybackslash}p{0.27\textwidth}@{}}
\toprule
Vertex $v$ & Incident edges in local order & Local ranks $k_{v,e}$ & Global ranks $g_{v,e}$\\
\midrule
$1$ & $\{1,2\},\{1,3\},\ldots,\{1,n\}$ & $1,2,\ldots,n-1$ & $1,2,\ldots,n-1$\\[0.6ex]
$2$ & $\{1,2\},\{2,3\},\ldots,\{2,n\}$ & $1,2,\ldots,n-1$ & $n,n+1,\ldots,2n-2$\\[0.6ex]
$\vdots$ & $\vdots$ & $\vdots$ & $\vdots$\\[0.6ex]
$i$ & $\{1,i\},\ldots,\{i-1,i\},\{i,i+1\},\ldots,\{i,n\}$ & $1,2,\ldots,n-1$ & $(i-1)(n-1)+1,\ldots,i(n-1)$\\[0.6ex]
$\vdots$ & $\vdots$ & $\vdots$ & $\vdots$\\[0.6ex]
$n$ & $\{1,n\},\{2,n\},\ldots,\{n-1,n\}$ & $1,2,\ldots,n-1$ & $(n-1)^2+1,\ldots,M$\\
\bottomrule
\end{tabularx}
\end{table}

% A8.4.3: remove repeated equivalent rank descriptions (retained in blue)
For every vertex $v$, ordering the other endpoints increasingly while omitting $v$ gives the same set of local ranks $\{1,\ldots,n-1\}$. Equivalently, for a neighbor $u\ne v$,
\[
k_{v,\{u,v\}}=
\begin{cases}
u,&u<v,\\
u-1,&u>v.
\end{cases}
\]
In the general row of Table~\ref{tab:ranks}, ranks $1,\ldots,i-1$ correspond to the lower-numbered neighbors of $i$, and ranks $i,\ldots,n-1$ correspond to its higher-numbered neighbors. Global ranks depend on the vertex: the incidences of vertex $v$ occupy the consecutive block
\[
\{(v-1)(n-1)+1,\ldots,v(n-1)\}.
\]
These $n$ blocks partition $\{1,\ldots,M\}$. 

\begin{comment}
% A8.4.3: concrete edge-versus-incidence example
\begin{table}[htbp]
{\color{red}
\centering
\caption{Incidence ranks for $n=3$, where $m=2$, $E=3$, and $M=6$}\label{tab:ranks-example}
\small
\setstretch{1.1}
\begin{tabularx}{\textwidth}{@{}cYccc@{}}
\toprule
Vertex & Incident edges in local order & Edge indices & Local ranks & Global ranks\\
\midrule
$1$ & $\{1,2\},\ \{1,3\}$ & $1,2$ & $1,2$ & $1,2$\\[0.7ex]
$2$ & $\{1,2\},\ \{2,3\}$ & $1,3$ & $1,2$ & $3,4$\\[0.7ex]
$3$ & $\{1,3\},\ \{2,3\}$ & $2,3$ & $1,2$ & $5,6$\\
\bottomrule
\end{tabularx}
\medskip
\parbox{\textwidth}{\footnotesize For example, edge $\{1,3\}$ has index $t_{\{1,3\}}=2$, but its incidence ranks are $g_{1,\{1,3\}}=2$ and $g_{3,\{1,3\}}=5$. Its job pairs are $(A_2,B_2)$ and $(A_5,B_5)$.}
\par}
\end{table}
\end{comment}

\paragraph{Charge coefficients and scaling.}
Define
\[
P_0=2M,\qquad L_0=2,\qquad q_e=M-2t_e\quad(e\in\mathcal E),
\]
and, for every vertex $v\in\mathcal V$, let
\begin{equation}
f_v=\sum_{e\in\mathcal E_v}\Delta_e
 \left[P_0\left(q_e+\frac32\right)+L_0(g_{v,e}-1)\right].\label{eq:fees}
\end{equation}
% A8.4.5: conceptual explanation of the fee formula
The coefficients are chosen so that, in the cost calculation for the normalized schedules, all terms depending on only one endpoint choice cancel out. What remains is a constant baseline together with a pairwise saving that occurs exactly when the two endpoint choices differ. Proposition~\ref{prop:coefficients} below verifies this cancellation explicitly after scaling.

Every $q_e$ is a nonnegative even integer. This, together with the fact that $P_0$ is even and $g_{v,e}$ is a positive integer, implies that all $f_v$ are positive integers. Write
\[
f_{\min}=\min_{1\le v\le n}f_v
\]
and set
\begin{align}
U&=1+V(2M+1)P_0+\sum_{v=1}^{n} f_v,\label{eq:U}\\
G&=\sum_{\substack{e=\{a,b\}\in\mathcal E\\a<b}}(U+h_e)
 \left[P_0(2q_e+3)+L_0(g_{a,e}+g_{b,e}-2)\right],\label{eq:G}\\
Q&=\left\lfloor\frac{G}{2P_0f_{\min}}\right\rfloor+1.\label{eq:Q}
\end{align}
In \eqref{eq:G}, the sum defining $G$ includes every original and added edge of $\mathcal E$ exactly once. Finally, let
\begin{equation}
P=QP_0,\quad L=QL_0,\quad F_v=Qf_v\ (v\in\mathcal V),\quad
\FS=\sum_{v=1}^{n}F_v,\quad \UB=QG.
\label{eq:actual}
\end{equation}
The sum defining $\FS$ includes one charge weight for every vertex. The parameters are defined in the displayed order; no optimum or unknown subset is used in calculating them.

\paragraph{Block boundaries and jobs.}
The $n$ vertex blocks are time intervals and therefore have $n+1$ boundary times. Define
\begin{equation}
b_s=(s-1)((m+2)P+1)\quad(1\le s\le n+1),\qquad H=b_{n+1}.
\label{eq:blocks}
\end{equation}
Vertex $v\in\mathcal V$ has block $[b_v,b_{v+1})$. Thus $b_1=0$, and $b_{n+1}=H$ is the end of the last vertex block, not an additional vertex. It is also the due date of the last guard and the starting time of the normalized tail.

For each incidence $g=g_{v,e}=(v-1)m+k_{v,e}$, introduce two jobs $A_g,B_g$. Add a charge job $Z_v$ and a guard job $R_v$ for each vertex $v\in\mathcal V$. Every job has processing time $P$. Their release dates, due dates, and weights are specified in Table~\ref{tab:jobs}. In particular, the two jobs $A_g,B_g$ associated with incidence $g$ have common due date $d_g=H-(g-1)L$. There are $2M+2n=2n^2$ jobs.

\begin{table}[htbp]
\centering
\caption{Job data in the reduction; every job has processing time $P$}\label{tab:jobs}
\setstretch{1.15}
\small
\begin{tabular}{lllll}
\toprule
Job & Release & Due date & Weight & Number\\
\midrule
$A_g$ & $b_v+(k-1)P$ & $H-(g-1)L$ & $U+\ell_e$ & $M$\\
$B_g$ & $b_v+(k-1)P+1$ & $H-(g-1)L$ & $U+h_e$ & $M$ \\
$Z_v$ & $b_v+mP$ & $b_v+(m+1)P$ & $F_v$ & $n$\\
$R_v$ & $b_v+(m+1)P+1$ & $b_{v+1}$ & $\UB+1$ & $n$\\
\bottomrule
\end{tabular}
\vskip 0.15cm
\parbox{0.65\textwidth}{\footnotesize In the $A_g$ and $B_g$ rows, $v\in\mathcal V$, $e\in\mathcal E_v$, $k=k_{v,e}\in\{1,\ldots,m\}$, and $g=g_{v,e}=(v-1)m+k\in\{1,\ldots,M\}$. Each incidence supplies one $A_g$-job and one $B_g$-job. In the charge and guard rows, $v=1,\ldots,n$.}
\end{table}

The target threshold is
\begin{equation}
K_* = \UB-\frac P2\Gamma\kappa.
\label{eq:threshold}
\end{equation}
All job data and the bound $\UB$ depend only on the graph. The MAX-CUT target $\kappa$ enters the nontrivial scheduling construction solely through the decision threshold $K_*$.

\subsection{Bounds needed for the schedule argument} \label{sec:construction-bounds}
\begin{lemma}[Parameter bounds]\label{lem:parameter-bounds}
The parameters defined in Section~\ref{sec:construction-data} satisfy
\begin{align}
P&=ML>1,\label{eq:length}\\
\left(\min_{1\le i\le n}F_i\right)(2P+1)&>\UB,\label{eq:chargeforce}\\
UL&>V(2M+1)P+\FS,\label{eq:exchangebound}\\
PU&>\FS.\label{eq:addbound}
\end{align}
\end{lemma}
\begin{proof}
Recall that $P_0=2M$, $L_0=2$, $Q\ge1$, and $M=n(n-1)\ge2$. Therefore, from (\ref{eq:actual}),
\[
P=QP_0=2QM=MQL_0=ML>1,
\]
which proves \eqref{eq:length}.

For \eqref{eq:chargeforce}, Equation~\eqref{eq:Q} and $f_{\min}>0$ give
\begin{equation} \label{eq:Qdefined}
Q=\left\lfloor\frac{G}{2P_0f_{\min}}\right\rfloor+1
>\frac{G}{2P_0f_{\min}}.
\end{equation}
Hence, using $\min_{1\le i\le n}F_i=Qf_{\min}$, $P=QP_0$, and $\UB=QG$ from the scaling definitions \eqref{eq:actual},
\[
\begin{aligned}
\left(\min_{1\le i\le n}F_i\right)(2P+1)
&=Qf_{\min}(2QP_0+1)\\
&>2Q^2P_0f_{\min}\\
&>QG, \quad {\rm from} \; (\ref{eq:Qdefined}) \\
&=\UB.
\end{aligned}
\]
 Let $\delta=2P_0f_{\min}Q-G$. By the definition of $Q$ given in \eqref{eq:Q}, we have $1\le \delta\le 2P_0f_{\min}$. The exact margin is
\begin{equation}
\begin{aligned}
\left(\min_{1\le i\le n}F_i\right)(2P+1)-\UB
&=Q\bigl[2P_0f_{\min}Q+f_{\min}-G\bigr]\\
&=Q(\delta+f_{\min})>0.
\end{aligned}
\label{eq:charge-margin}
\end{equation}

For \eqref{eq:exchangebound}, Equation~\eqref{eq:U} gives
\[
V(2M+1)P_0+\sum_{v=1}^{n}f_v=U-1.
\]
Using $L=2Q$, $P=QP_0$ and $\FS=Q\sum_{i=1}^{n}f_i$
from (\ref{eq:actual}), we obtain
\[
V(2M+1)P+\FS=Q(U-1)<2QU=UL.
\]
The exact margin, which is used below, is
\begin{equation}
\begin{aligned}
UL-V(2M+1)P-\FS
&=Q\left[2U-V(2M+1)P_0-\sum_{i=1}^{n} f_i\right]\\
&=Q\left[2+V(2M+1)P_0+\sum_{i=1}^{n} f_i\right]\\
&=Q(U+1)>0.
\end{aligned}
\label{eq:exchange-margin}
\end{equation}

Finally, Equation~\eqref{eq:U} gives $U>\sum_{i=1}^{n}f_i$. Since $P_0\ge1$ and $Q>0$,
\[
PU=QP_0U\ge QU>Q\sum_{i=1}^{n}f_i=\FS,
\]
from (\ref{eq:U}) and (\ref{eq:actual}), and proving \eqref{eq:addbound}. Equivalently, the exact margin is $PU-\FS=Q(P_0U-\sum_{i=1}^{n}f_i)>0$. %These calculations use positivity and the parameter definitions; they do not use the graph-specific formula \eqref{eq:graphcoefficients}.
\end{proof}

% A8.4.3: avoid a further repetition of the rank construction
The latest release among the $A/B$-jobs is that of $B_M$, whose local rank is $m$ at vertex $n$.  By the $B_g$ row of Table~\ref{tab:jobs} and the block definition \eqref{eq:blocks}, its release date is
\begin{equation}
r_{\max}^{AB}
=b_n+(m-1)P+1
=H-3P<H.
\label{eq:abrelease}
\end{equation}
In particular, every $A/B$-job is available at time $H$, and the latest such release plus $P$ is $H-2P$.  By Table~\ref{tab:jobs}, the earliest due date occurs at $g=M$; using $(M-1)L=P-L$ from \eqref{eq:length}, it is
\begin{equation}
H-(M-1)L=H-P+L>H-P.
\label{eq:duebound}
\end{equation}
Thus $d_j\ge r_j+P$ for every $A/B$-job, and equality holds for the charge and guard jobs.  Indeed, \eqref{eq:abrelease} and \eqref{eq:duebound} give
$r_j+P\le H-2P<H-P<d_j$ for every $A/B$-job, while the equalities for
charge and guard jobs follow from Table~\ref{tab:jobs} and
\eqref{eq:blocks}.

 We next give an explicit reference schedule with total weighted tardiness
$\overline K$. This schedule is not claimed to be optimal; its role is to
show that the constructed instance has a feasible schedule of cost at most
$\overline K$, and hence  the class of schedules considered in the
normalization argument is nonempty. It also gives the baseline cost from
which the later cut-dependent savings are measured.

Let $\sigma$ be a permutation of $\{1,\ldots,M\}$ such that
\[
w_{B_{\sigma(1)}}\ge\cdots\ge w_{B_{\sigma(M)}},
\]
with ties broken arbitrarily. Its job sequence is
\[
\begin{aligned}
(&A_1,\ldots,A_m,Z_1,R_1,\;
A_{m+1},\ldots,A_{2m},Z_2,R_2,\;\ldots,\\
&A_{(n-1)m+1},\ldots,A_M,Z_n,R_n,\;
B_{\sigma(1)},\ldots,B_{\sigma(M)}).
\end{aligned}
\]
The associated start times are
\[
\begin{aligned}
s_{A_{(v-1)m+k}}&=b_v+(k-1)P
&& (1\le v\le n,\ 1\le k\le m),\\
s_{Z_v}&=b_v+mP
&& (1\le v\le n),\\
s_{R_v}&=b_v+(m+1)P+1
&& (1\le v\le n),\\
s_{B_{\sigma(t)}}&=H+(t-1)P
&& (1\le t\le M).
\end{aligned}
\]
Thus there is one unit of idle time between the completion of each $Z_v$ and the start of $R_v$. Each guard finishes at $b_{v+1}$, and the last guard finishes at $H$. All $A$-jobs, charge jobs, and guard jobs are on time. The final suffix after time $H$ consists of the $B$-jobs; informally, this is
the tail of the reference schedule. 
For each edge $e=\{a,b\}$, the two tail jobs are
$B_{g_{a,e}}$ and $B_{g_{b,e}}$, both of weight $U+h_e$
(Table~\ref{tab:jobs}). The weight intervals increase with $t_e$
by \eqref{eq:intervals}, so exactly $2(E-t_e)=M-2t_e=q_e$
tail jobs have greater weight. Hence these two jobs occupy tail
positions $q_e+1$ and $q_e+2$, with completion times
$H+(q_e+1)P$ and $H+(q_e+2)P$, in either order.
Their due dates are $H-(g_{a,e}-1)L$ and $H-(g_{b,e}-1)L$.
Both jobs complete after $H$ and have due dates at most $H$, so
$T_j=C_j-d_j$ for each of them. Their total tardiness is therefore
\[
\begin{aligned}
& w_{B_{g_{a,e}}}T_{B_{g_{a,e}}}
  +w_{B_{g_{b,e}}}T_{B_{g_{b,e}}}\\
&\quad=(U+h_e)\Bigl[2H+(2q_e+3)P
  -\bigl(2H-L(g_{a,e}+g_{b,e}-2)\bigr)\Bigr]\\
&\quad=(U+h_e)\bigl[P(2q_e+3)+L(g_{a,e}+g_{b,e}-2)\bigr].
\end{aligned}
\]
The first equality subtracts the sum of their due dates from the sum
of their completion times; their order is immaterial because their
weights are equal.

% A8.4.4: consolidate reference-bound safeguards without removing the witness

Using (\ref{eq:G}) and (\ref{eq:actual}) gives
\[
\begin{aligned}
W_{\mathrm{ref}}
&=\sum_{\substack{e=\{a,b\}\in\mathcal E\\a<b}}
 (U+h_e)\bigl[P(2q_e+3)+L(g_{a,e}+g_{b,e}-2)\bigr]\\[0.6ex]
&=\sum_{\substack{e=\{a,b\}\in\mathcal E\\a<b}}
 (U+h_e)\bigl[QP_0(2q_e+3)+QL_0(g_{a,e}+g_{b,e}-2)\bigr]\\[0.6ex]
&=\sum_{\substack{e=\{a,b\}\in\mathcal E\\a<b}}
 Q(U+h_e)\bigl[P_0(2q_e+3)+L_0(g_{a,e}+g_{b,e}-2)\bigr]\\[0.6ex]
&=Q\sum_{\substack{e=\{a,b\}\in\mathcal E\\a<b}}
 (U+h_e)\bigl[P_0(2q_e+3)+L_0(g_{a,e}+g_{b,e}-2)\bigr]\\[0.6ex]
&=QG\\
&=\overline K.
\end{aligned}
\]
This complete schedule processes every job exactly once and proves that the optimal objective value of the constructed instance $\OPT\le\UB$. The smaller value $K_*=\UB-(P/2)\Gamma\kappa$ is the decision threshold, not the cost of this reference schedule and not necessarily the optimum. Cut-dependent schedules may meet or improve on $K_*$ by choosing different modes at different vertices. %The value $K_*$ is the separate decision threshold, not the cost of the reference schedule.

\section{Schedule normalization}\label{sec:normalization}
 This section proves that every feasible schedule of cost at most $\UB$ can be
transformed, in polynomial time and without increasing cost, into a highly
structured canonical schedule. Consequently, for the decision threshold $K_*\le\UB$, it suffices to analyze
canonical schedules.  We first determine the positions of the guard and charge jobs and characterize the feasible sets of early $A/B$-jobs. We then use exchange arguments to obtain exactly one early job from each pair $(A_g,B_g)$ and to enforce a consistent $A$- or $B$-mode at every vertex. These steps culminate in a normalization theorem showing that every bounded-cost schedule, including one with noninteger start times, can be transformed in polynomial time and without increasing cost into a canonical schedule,  as established in Proposition~\ref{prop:normalization}. This structure will allow Section~\ref{sec:identity} to encode a schedule by one binary choice per vertex and calculate its cost edge by edge. It applies to every schedule meeting the decision threshold because $K_*\le\UB$ by \eqref{eq:threshold}.

Figure~\ref{fig:proof-overview} summarizes the normalization steps and their
proof references, followed by the cost identity and threshold equivalence.

% Embedded proof overview (formerly ProofOverviewFigure.tex).
\begin{figure}[p]
\centering
\begingroup
\setstretch{1}
\color{black}
% Slight scaling keeps the overview and its caption within the text height.
\begin{tikzpicture}[
  x=1cm,y=1cm,scale=.97,transform shape,
  every node/.style={font=\fontsize{9}{10.7}\selectfont},
  po box/.style={draw=black!65,line width=.5pt,rounded corners=1.5mm,
    fill=black!2,text width=10.05cm,align=center,inner xsep=2mm,
    inner ysep=1.8mm},
  po ref/.style={anchor=north west,text width=5.08cm,align=left,
    inner sep=0pt,font=\fontsize{8.5}{10.2}\selectfont},
  po flow/.style={-{Latex[length=1.6mm,width=1.2mm]},line width=.65pt},
  po mode/.style={po box,text width=4.51cm,minimum height=1.72cm},
  po result/.style={po box,fill=black!7,line width=.8pt},
  po cost/.style={po flow,dash pattern=on 1pt off 1.2pt}
]
\coordinate (po legend) at (11.05,0);
\node[anchor=north west,inner sep=0pt,font=\bfseries\fontsize{12}{14}\selectfont]
  at (0,0) {Overview of the proof};
\node[anchor=north west,inner sep=0pt,font=\bfseries\fontsize{10}{12}\selectfont]
  at (11.05,0) {Legend: proof references};
\node[anchor=north west,inner sep=0pt,text width=10.45cm,align=left,
  font=\fontsize{8.5}{10.2}\selectfont]
  at (0,-.48) {Steps 1--7 preserve feasibility and never increase total weighted tardiness.};

\node[po box,fill=white,minimum height=.90cm,anchor=north]
  (po input) at (5.225,-1.02)
  {\textbf{Any feasible schedule with cost $\le\UB$}\\
   Start times may be noninteger; the job permutation is sufficient.};
\node[po ref] at ($(po input.north -| po legend)+(0,-.02)$)
  {\textbf{Input schedule and normalization bound.} \\ 
  $K_*\le\UB$: \eqref{eq:threshold}.};

\node[po box,below=2.6mm of po input,minimum height=1.24cm] (po one)
  {\textbf{1\quad Left shift the given permutation}\\[1pt]
   $\widehat s_{\pi(k)}=\max\{r_{\pi(k)},\widehat C_{k-1}\}$,\quad
   $\widehat C_k=\widehat s_{\pi(k)}+P$,\quad $\widehat C_0=0$.\\
   Every start becomes integral; no completion time increases.};
\node[po ref] at ($(po one.north -| po legend)+(0,-.05)$)
{\textbf{1\quad Proposition~\ref{prop:normalization}, proof}\\
Applied through Lemma~\ref{lem:exchange}; the earliest-permutation
recurrence is \eqref{eq:leftshift}.};

\node[po box,below=2.6mm of po one,minimum height=1.26cm] (po two)
  {\textbf{2\quad Fix the guard and charge positions}\\[1pt]
   $s_{R_i}=r_{R_i}=b_i+(m+1)P+1$.\\
   $s_{Z_i}\in\{b_i+mP,\ b_i+mP+1\}$;\quad
   $w_{Z_i}T_{Z_i}\in\{0,F_i\}$.};
\node[po ref] at ($(po two.north -| po legend)+(0,-.05)$)
  {\textbf{2\quad Lemma~\ref{lem:controls}: Control placement}\\
   Integrality and the cost bound force these positions.
   Each early window holds at most $m$ jobs.};

\node[po box,below=2.6mm of po two,minimum height=1.26cm] (po three)
  {\textbf{3\quad Compact the tail before any exchanges}\\[1pt]
   Pack the tail consecutively from $H=b_{n+1}$ in its current order.\\
   $H+P\le C\le H+2MP$,\quad
   $P\le C-d_g\le(2M+1)P$.};
\node[po ref] at ($(po three.north -| po legend)+(0,-.05)$)
  {\textbf{3\quad Compaction before Lemma~\ref{lem:matching}}\\
   That paragraph proves \eqref{eq:tailbound}; \eqref{eq:abrelease} puts
   every tail release before $H$.};

\node[po box,below=2.6mm of po three,minimum height=1.55cm] (po four)
  {\textbf{4\quad Repair early counts, from the largest defective rank}\\[1pt]
   Let $c_g=|\{A_g,B_g\}\cap S|$, where $S$ is the early set.\\
   Choose the largest $g$ with $c_g\ne1$; necessarily $c_g=0$.\\
   Insert $A_g$ early, remove at most one lower-rank early job,\\
   and rematch the early set. Repeat until $c_g=1$ for every $g$.};
\node[po ref] at ($(po four.north -| po legend)+(0,-.03)$)
  {\textbf{4\quad Lemma~\ref{lem:exchange}: One early job per incidence}\\
   Lemma~\ref{lem:matching} supplies matching via the suffix condition
   \eqref{eq:suffix}:\\
   $\sum_{h=g}^{M}c_h\le M-g+1$.\\
   At most $M$ repairs; cost decreases.};

\node[po box,below=2.6mm of po four,minimum height=1.43cm] (po five)
  {\textbf{5\quad Recover each incidence's own window and position}\\[1pt]
   All $M$ early slots are filled. Their rank permutation satisfies\\
   $\rho_g\le g$, hence $\rho_g=g$ for every $g$. Left shift again:\\
   $s_k=b_i+(k-1)P+\varepsilon_k$,\quad
   $\varepsilon_k=\max\{\varepsilon_{k-1},z_k\}$,\quad $\varepsilon_0=0$.};
\node[po ref] at ($(po five.north -| po legend)+(0,-.04)$)
  {\textbf{5\quad Lemma~\ref{lem:consistency}, proof: rank forcing}\\
   Capacity and releases force $\rho_g=g$.\\
   Here $s_k$ is the local early start, $z_k=0$ for a selected $A$
   and $z_k=1$ for a selected $B$;
   $\varepsilon_k\in\{0,1\}$.};

\node[po box,below=2.6mm of po five,minimum height=1.22cm] (po six)
  {\textbf{6\quad Make each vertex consistent}\\[1pt]
   Mixed $A/B$ vertex $\longrightarrow$ all-$B$ vertex.\\
   Swap each early $A_g$ with its tail $B_g$; pure vertices keep their mode.};
\node[po ref] at ($(po six.north -| po legend)+(0,-.04)$)
  {\textbf{6\quad Lemma~\ref{lem:consistency}: Vertex consistency}\\
   A mixed vertex already pays $F_i$. The charge cost stays fixed,
   while the tail saving in \eqref{eq:mixedsaving} is strictly positive.};

\node[po mode,anchor=north] (po a)
  at ($(po six.south)+(-2.67,-.36)$)
  {\textbf{$A$-mode: $x_i=1$}\\[2pt]
   $A_g$ early; $B_g$ in the tail.\\
   $s_k=b_i+(k-1)P$.\\
   Charge cost $0$; $T_{Z_i}=0$.\\
   One unit idle before $R_i$.};
\node[po mode,anchor=north] (po b)
  at ($(po six.south)+(2.67,-.36)$)
  {\textbf{$B$-mode: $x_i=0$}\\[2pt]
   $B_g$ early; $A_g$ in the tail.\\
   $s_k=b_i+(k-1)P+1$.\\
   Charge cost $F_i$; $T_{Z_i}=1$.\\
   One unit idle after $b_i$.};
\node[po ref] at ($(po a.north -| po legend)+(0,-.05)$)
  {\textbf{Two resulting modes, independently by vertex.}\\
   In these boxes, $g=(i-1)m+k$ and $1\le k\le m$.
   Both modes have\\
   $s_{R_i}=b_i+(m+1)P+1$\\
   and $C_{R_i}=b_{i+1}$: the guard resets the next block.};

\coordinate (po merge) at ($(po a.south)!.5!(po b.south)$);
\node[po result,anchor=north,minimum height=1.50cm]
  (po seven) at ($(po merge)+(0,-.36)$)
  {\textbf{7\quad Compact and sort the tail: canonical schedule}\\[1pt]
   Exactly $M$ tail jobs run consecutively from $H$ in nonincreasing $w_j$.\\
   Each vertex now has one bit $x_i$; its charge contributes $F_i(1-x_i)$.\\
   Early $A/B$-jobs and all guards contribute zero.};
\node[po ref] at ($(po seven.north -| po legend)+(0,-.04)$)
  {\textbf{7\quad Proposition~\ref{prop:normalization}: Normalization}\\
   For the fixed tardy tail set, minimize $\sum w_jC_j$ by sorting weights.
   The entire transformation is polynomial and does not assume optimality.};

\node[po result,below=3mm of po seven,minimum height=1.43cm] (po eight)
  {\textbf{8\quad Evaluate the canonical schedule by its cut}\\[1pt]
   $\displaystyle C(x)=\UB-\frac P2
     \sum_{1\le a<b\le n}\Delta_{\{a,b\}}|x_a-x_b|$.\\[2pt]
   $C(x)\le K_*\ \Longleftrightarrow\ c_{\mathcal G}(x)\ge\kappa$
   \quad for the graph-specific coefficients.};
\node[po ref] at ($(po eight.north -| po legend)+(0,-.04)$)
  {\textbf{8\quad Proposition~\ref{prop:coefficients}: Cost identity}\\
   Equation~\eqref{eq:cutcost}; then \textbf{Theorem~\ref{thm:main}} gives
   the threshold equivalence. Polynomial numerical bounds and membership
   in NP complete strong NP-completeness.};

\foreach \a/\b in {input/one,one/two,two/three,three/four,four/five,five/six}
  \draw[po flow] (po \a.south) -- (po \b.north);
\coordinate (po split) at ($(po six.south)+(0,-.18)$);
\draw[line width=.65pt] (po six.south) -- (po split);
\draw[po flow] (po split) -| (po a.north);
\draw[po flow] (po split) -| (po b.north);
\coordinate (po join) at ($(po merge)+(0,-.18)$);
\draw[line width=.65pt] (po a.south) |- (po join);
\draw[line width=.65pt] (po b.south) |- (po join);
\draw[po flow] (po join) -- (po seven.north);
\draw[po cost] (po seven.south) -- (po eight.north);

\node[anchor=north west,inner sep=0pt,text width=16.2cm,align=left,
  font=\fontsize{8}{9.5}\selectfont]
  at ($(po eight.south -| 0,0)+(0,-.25)$)
  {Solid arrows show normalization; branches show possible vertex modes;
   the dashed arrow applies the cost identity.};
\end{tikzpicture}
\caption{Normalization of a bounded-cost schedule. Steps 1--7 preserve
feasibility and do not increase total weighted tardiness; Step 8 applies the
cost identity and the MAX-CUT threshold equivalence. The legend identifies
the supporting results.}
\label{fig:proof-overview}
\endgroup
\end{figure}
\afterpage{\clearpage}

\begin{lemma}[Control placement]\label{lem:controls}
In any integral feasible schedule of cost at most $\UB$, for every $i\in\mathcal V$, guard job $R_i$ starts at its release date, while charge job $Z_i$ starts at either $b_i+mP$ or $b_i+mP+1$ and has weighted-tardiness cost zero or $F_i$, respectively.
\end{lemma}
\begin{proof}
By Table~\ref{tab:jobs} and \eqref{eq:blocks}, $d_{R_i}=r_{R_i}+P$ and $w_{R_i}=\UB+1$. Thus, if a guard job were late in the integral schedule, it would cost at least $\UB+1$. Every other cost is nonnegative, so this would exceed the schedule's cost bound. Thus it occupies its prescribed interval.

Let $s=b_i+(m+1)P+1$ be its start. Charge $Z_i$ has release $s-P-1$ and due date $s-1$. If the charge does not finish by $s$, nonpreemption and the fixed guard imply it cannot start before $s+P$ and cannot finish before $s+2P$. Its tardiness would be at least $2P+1$, contradicting \eqref{eq:chargeforce}. Integral feasibility before $s$ leaves exactly the two stated charge starts.
\end{proof}

For an integral schedule of cost at most $\UB$, consider the interval in
vertex block $i$ before the charge job starts. Lemma~\ref{lem:controls}
places the charge start at either $b_i+mP$ or $b_i+mP+1$, so this interval
has length at most $mP+1$ and hence fits at most $m$ $A/B$-jobs. The gap after a charge and before its guard is at most one unit, less than $P$. The last guard occupies $[H-P,H)$, so no $A/B$-job can straddle $H$. Consequently, every $A/B$-job executed before $H$ is in an early window and completes no later than the last charge start $b_n+mP+1=H-2P$; by \eqref{eq:duebound} it is on time. Every other $A/B$-job starts at or after $H$ and is tardy, since Table~\ref{tab:jobs} gives $d_g\le H$. Call every $A$- or $B$-job processed at or after $H$ a \emph{tail job}. Thus the tail jobs are precisely the $A/B$-jobs not processed in an early window; charge and guard jobs are not included. By \eqref{eq:abrelease}, all tail jobs have release date at most $H-3P<H$. We may therefore compact them consecutively from $H$, retaining their order and leaving the early part and charge and guard jobs unchanged; this cannot increase cost. Since there are $2M$ $A/B$-jobs in total, there are at most $2M$ tail jobs, so their completion times then satisfy
\begin{equation}
H+P\le C\le H+2MP,
\qquad P\le C-d_g\le(2M+1)P.
\label{eq:tailbound}
\end{equation} 
The completion-time bounds use the preceding compaction and the fact that
there are at most $2M$ tail jobs. For a tail $A/B$-job with incidence rank
$g$, Table~\ref{tab:jobs} gives
$C-d_g=(C-H)+(g-1)L$. By \eqref{eq:length},
$0\le(g-1)L\le(M-1)L=P-L<P$, which gives the stated bounds on the job's tardiness $C-d_g$.
These bounds do not hold for an arbitrary uncompressed tail; compaction supplies the finite horizon before any exchange is used. The same bounds remain valid if some of these tail slots are later left empty. We use this observation during the exchanges below.

\begin{lemma}[Nested slot matching]\label{lem:matching}
Under the charge and guard placements of Lemma~\ref{lem:controls}, let $S$
be a selected set of early $A/B$-jobs and let $c_g\in\{0,1,2\}$ count its
members from the job pair $(A_g,B_g)$, for $1\le g\le M$. If $S$ is
schedulable before the charge jobs, then necessarily
\begin{equation}
\sum_{h=g}^{M}c_h\le M-g+1\qquad(1\le g\le M).
\label{eq:suffix}
\end{equation}
Conversely, every selected set $S$ whose counts satisfy \eqref{eq:suffix}
can be scheduled in the following early slots: slot $g=(i-1)m+k$ starts at
\begin{equation}
b_i+(k-1)P+1,
\qquad 1\le i\le n,\quad 1\le k\le m,
\label{eq:grid}
\end{equation}
with all charges one unit late and all guards at their release dates.
\end{lemma}  
\begin{proof}
From the release date $b_i+(k-1)P$ associated with incidence rank
$g=(i-1)m+k$ onward, its vertex window has length at most
$(m-k+1)P+1$ by Lemma~\ref{lem:controls}. It therefore fits at most
\[
\left\lfloor\frac{(m-k+1)P+1}{P}\right\rfloor=m-k+1
\]
jobs, since $P>1$ by \eqref{eq:length}. Each later window fits at most
$m$. Neither member of a job pair ranked at least $g$ can run before this
release date, by the release dates in Table~\ref{tab:jobs} and the ordering
of the global incidence ranks. The total capacity is at most
$(m-k+1)+(n-i)m=M-g+1$. This proves \eqref{eq:suffix}; any older jobs
occupying these windows can only reduce the available capacity.

Conversely, the empty-set case uses no early $A/B$-job slots. Otherwise
let the chosen incidence ranks, with one entry for each selected job, be
$u_1\le\cdots\le u_s$. The suffix inequality at $u_t$ gives
$s-t+1\le M-u_t+1$, so $u_t\le M-s+t$. Assign job $t$ to slot number
$M-s+t$ in \eqref{eq:grid}, for $1\le t\le s$. Its incidence rank is at
most that slot number; thus its release date is satisfied, even if it is a
$B$-job. Slots do not overlap. The last slot in each vertex block completes
at its proposed charge-job start, and the charge job completes at its
guard's start. The latest early-slot completion in \eqref{eq:grid} is
$H-2P$, so \eqref{eq:duebound} makes every early $A/B$-job on time. By
Table~\ref{tab:jobs}, each charge job is one unit late; summing their
weights gives total charge contribution $\FS$ by \eqref{eq:actual}.
\end{proof}

\begin{lemma}[One early job per incidence]\label{lem:exchange}
Every feasible schedule of cost at most $\UB$ can be transformed in polynomial time into an integral feasible schedule of no greater cost with exactly one early member of each job pair $(A_g,B_g)$.
\end{lemma}
\begin{proof}
Apply \eqref{eq:leftshift} and, if needed, compact the tail from $H$. The cost remains at most $\UB$, Lemma~\ref{lem:controls} applies, and \eqref{eq:tailbound} holds. If every early count is one, we are done. Otherwise choose the largest index $g$ for which $c_g\ne1$. Higher counts are all one; \eqref{eq:suffix} therefore forces $c_g=0$, giving exactly one unit of slack in the suffix at $g$. Both members of this job pair are consequently in the tail.  Thus $A_g$ has
a current tail slot; let $C$ be its completion time in that slot.  Let $e=e(g)$. Insert $A_g$ into the early set, removing at most one job from a pair with a lower incidence rank as follows.

If no suffix with index below $g$ is tight in \eqref{eq:suffix}, insertion alone preserves every inequality. Otherwise let $t<g$ be the largest tight suffix index. All suffixes with indices in $(t,g]$ have positive integer slack.  The tight suffix at $t$ contains $M-t+1$ early jobs by \eqref{eq:suffix}; the suffix at $g$ contains $M-g$, because $c_g=0$ and $c_h=1$ for $h>g$. Subtracting gives the number of early jobs with ranks in $[t,g-1]$:
\[
(M-t+1)-(M-g)=g-t+1>0.
\]
Remove any such job $D$, with incidence rank $t\le h<g$. After insertion and removal, suffix counts at indices at most $h$ have net change zero, those between $h+1$ and $g$ increase by one but had slack, and all other suffixes are unchanged. Thus \eqref{eq:suffix} remains valid, and every count remains in $\{0,1,2\}$.

Rebuild the early part using Lemma~\ref{lem:matching}. Every rebuilt early $A/B$-job is on time, and all guards retain zero cost. If the total charge costs before and after rebuilding are $Z_{\mathrm{old}}$ and $Z_{\mathrm{new}}$, respectively, then $0\le Z_{\mathrm{old}},Z_{\mathrm{new}}\le\FS$, by Lemma~\ref{lem:controls} for the old charges, Lemma~\ref{lem:matching} for the rebuilt charges, and \eqref{eq:actual} for their total weight, so their increase is at most $\FS$. Leave guards and all tail slots fixed, except the old slot of $A_g$, which completes at $C$. If $D$ was removed, put it in that slot. Write its weight as $U+v_D$ with $1\le v_D\le V$,  by Table~\ref{tab:jobs}, \eqref{eq:intervals}, and \eqref{eq:V}.  The common due-date formula in Table~\ref{tab:jobs} gives $d_h=H-(h-1)L$ and $d_g=H-(g-1)L$. Since $h<g$, it follows that $d_h-d_g=(g-h)L\ge L$. The slot starts at or after $H$, and $D$ is released before $H$ by \eqref{eq:abrelease}. 
Moreover, $C-d_g$ lies in $[P,(2M+1)P]$ by \eqref{eq:tailbound}. Also,
$H+P\le C\le H+2MP$ for this tail slot, and
$0\le(h-1)L\le(M-1)L=P-L<P$ by \eqref{eq:length}; hence
$C-d_h=C-H+(h-1)L$ also lies in $[P,(2M+1)P]$. Thus replacement preserves feasibility, both costs are linear tardiness costs, and the tail saving is
\begin{align*}
&(U+\ell_e)(C-d_g)-(U+v_D)(C-d_h)\\
&\quad=U(d_h-d_g)+\ell_e(C-d_g)-v_D(C-d_h)\\
&\quad\ge UL-V(2M+1)P>\FS.
\end{align*}  The strict inequality is due to \eqref{eq:exchangebound} of Lemma~\ref{lem:parameter-bounds}.
The omitted term $\ell_e(C-d_g)$ is nonnegative, and the negative term is bounded using $v_D\le V$ and the tail horizon. Consequently, writing $J_{\mathrm{old}}$ and $J_{\mathrm{new}}$ for the total schedule costs, we obtain
\[
J_{\mathrm{old}}-J_{\mathrm{new}}
\ge UL-V(2M+1)P-(Z_{\mathrm{new}}-Z_{\mathrm{old}})
\ge Q(U+1)>0.
\] 
The last inequality uses $Z_{\mathrm{new}}-Z_{\mathrm{old}}\le\FS$ and the exact exchange margin \eqref{eq:exchange-margin}.
This bound covers simultaneous changes to all charges, not only to the two affected vertices.

If no job was removed, then after $A_g$ has been moved into the early part,
leave its old tail slot idle. By \eqref{eq:tailbound}, deleting that tail slot saves at least $P(U+\ell_e)$, while rebuilding charges adds at most $\FS$. Hence
\[
J_{\mathrm{old}}-J_{\mathrm{new}}
\ge P(U+\ell_e)-\FS
=Q\left[P_0(U+\ell_e)-\sum_{i=1}^{n} f_i\right]>0,
\]
The equality uses \eqref{eq:actual}, and strict positivity follows from \eqref{eq:addbound}.
In either case every job is still scheduled exactly once and total cost
strictly decreases. Hence the schedule still has cost at most $\UB$, and it
remains integral because all jobs are placed in integer slots. It also
continues to satisfy \eqref{eq:tailbound} for the following reason. The
exchange only replaces jobs in existing tail slots or leaves such slots
empty, so no occupied tail slot moves later than before. If a removed job
$D$ with incidence rank $h$ is placed in such a slot, then
$0\le(h-1)L\le(M-1)L=P-L<P$ by \eqref{eq:length}, so the bound
$P\le C-d_h\le(2M+1)P$ still applies. No tail sorting or compaction is required between exchanges. The largest defective rank has been corrected; every rank at least $g$ now has count one, and any remaining or newly created defect has smaller rank. Repeating therefore terminates after at most $M$ exchanges. Each exchange and slot reconstruction takes polynomial time and uses only the displayed job data and the current schedule. No optimality assumption is used.
\end{proof}

\begin{lemma}[Vertex consistency]\label{lem:consistency}
Every integral feasible schedule of cost at most $\UB$ with exactly one early member of each job pair $(A_g,B_g)$ can be transformed in polynomial time, without increasing cost, into a schedule in which every vertex $i$ uses one of two modes. In the $A$-mode, for every incidence $g$ of $i$, job $A_g$ is early and job $B_g$ is in the tail; charge job $Z_i$ completes at its due date and contributes zero to the objective. In the $B$-mode, every such $B_g$ is early and the corresponding $A_g$ is in the tail; $Z_i$ completes one unit after its due date and contributes $F_i$. In either mode, guard job $R_i$ starts at its release date.
\end{lemma}

Figure~\ref{fig:lemma5-modes} illustrates these two modes, including their
release and due dates, the one-unit idle interval, and the complementary
tail jobs.

% Embedded illustration of Lemma 5 (formerly Lemma5Figure.tex).
\begin{figure}[p]
\centering
\begingroup
\color{black}
\renewcommand{\baselinestretch}{1}\selectfont
\begin{tikzpicture}[x=1cm,y=1cm,font=\footnotesize]
\path[use as bounding box] (0,-14.25) rectangle (16.4,3.10);

\node[anchor=west,font=\small] at (0,2.89)
  {\textbf{Vertex block $i$:} $g_k=(i-1)m+k$, $1\leq k\leq m$; every job has duration $P$.};
% Legend: colors are reinforced by patterns, borders and explicit labels.
\draw[lf5/A] (0,2.20) rectangle (.40,2.49);
\node[anchor=west] at (.48,2.345) {$A$-job};
\draw[lf5/B] (2.10,2.20) rectangle (2.50,2.49);
\node[anchor=west] at (2.58,2.345) {$B$-job};
\draw[lf5/charge] (4.20,2.20) rectangle (4.60,2.49);
\node[anchor=west] at (4.68,2.345) {charge $Z_i$};
\draw[lf5/guard] (7.00,2.20) rectangle (7.40,2.49);
\node[anchor=west] at (7.48,2.345) {guard $R_i$};
\node[lf5/release] at (10.05,2.345) {};
\node[anchor=west] at (10.20,2.345) {release};
\node[lf5/due] at (12.35,2.345) {};
\node[anchor=west] at (12.50,2.345) {due date};

% Display geometry: P=2.5 cm, displayed one-unit offset=.55 cm.
% The .7 cm ellipsis omits incidences k=3,...,m-1; it has no time scale.
% These display lengths are NOT a numerical instance of the reduction.
\begin{scope}
  \node[anchor=west,font=\small\bfseries] at (0,1.92)
    {(a) $A$-mode: all $A$-jobs early};
  \node[anchor=east] at (16.25,1.92) {$T_{Z_i}=0,\quad w_{Z_i}T_{Z_i}=0$};
  \foreach \x in {1,3.5,6,6.7,9.2,11.7,12.25,14.75}
    \draw[lf5/guide] (\x,0)--(\x,1.10);
  \LFiveJob{A}{1}{3.5}{$A_{g_1}$}
  \LFiveJob{A}{3.5}{6}{$A_{g_2}$}
  \node at (6.35,.70) {$\cdots$};
  \LFiveJob{A}{6.7}{9.2}{$A_{g_m}$}
  \LFiveJob{charge}{9.2}{11.7}{$Z_i$}
  \draw[lf5/idle] (11.7,.35) rectangle (12.25,1.05);
  \LFiveGuardJob{12.25}{14.75}{$R_i$}
  \foreach \x in {1,3.5,6.7,9.2,12.25} \LFiveRelease{\x}
  \LFiveDue{11.7}\LFiveDue{14.75}
  \draw[lf5/brace] (1,1.25)--(3.5,1.25) node[midway,above=3pt] {$P$};
  \draw[lf5/brace] (11.7,1.25)--(12.25,1.25) node[midway,above=3pt] {$1$};
  \node[anchor=east] at (11.45,1.28) {$C_{Z_i}=d_{Z_i}$};
  \node[anchor=east] at (16.25,1.28) {$C_{R_i}=d_{R_i}$};
  \draw[lf5/axis] (.55,0)--(15.45,0);
  \node[anchor=west] at (15.48,0) {$t-b_i$};
  \LFiveTick{1}{-.15}{$0$}
  \LFiveTick{3.5}{-.15}{$P$}
  \LFiveTick{6}{-.15}{$2P$}
  \LFiveTick{6.7}{-.66}{$(m-1)P$}
  \LFiveTick{9.2}{-.15}{$mP$}
  \LFiveTick{11.7}{-.15}{$(m+1)P$}
  \LFiveTick{12.25}{-.66}{$(m+1)P+1$}
  \LFiveTick{14.75}{-.15}{$(m+2)P+1$}
  \draw[black!45,line width=.35pt] (6.7,-.07)--(6.7,-.62);
  \draw[black!45,line width=.35pt] (12.25,-.07)--(12.25,-.62);
  \node[anchor=west] at (0,-1.16)
    {Idle: $[d_{Z_i},r_{R_i})$, length $1$. \quad Early jobs and guard: zero tardiness.};
  \node[anchor=west] at (0,-1.57)
    {Tail members from this vertex: $B_{g_1},\ldots,B_{g_m}$, with weights $U+h_{e(g_k)}$.};
\end{scope}

\draw[black!18] (0,-1.82)--(16.4,-1.82);
\begin{scope}[yshift=-4.05cm]
  \node[anchor=west,font=\small\bfseries] at (0,1.92)
    {(b) $B$-mode: all $B$-jobs early};
  \node[anchor=east] at (16.25,1.92) {$T_{Z_i}=1,\quad w_{Z_i}T_{Z_i}=F_i$};
  \foreach \x in {1,1.55,4.05,6.55,7.25,9.2,9.75,11.7,12.25,14.75}
    \draw[lf5/guide] (\x,0)--(\x,1.10);
  \draw[lf5/idle] (1,.35) rectangle (1.55,1.05);
  \LFiveJob{B}{1.55}{4.05}{$B_{g_1}$}
  \LFiveJob{B}{4.05}{6.55}{$B_{g_2}$}
  \node at (6.90,.70) {$\cdots$};
  \LFiveJob{B}{7.25}{9.75}{$B_{g_m}$}
  \LFiveJob{charge}{9.75}{12.25}{$Z_i$}
  \LFiveGuardJob{12.25}{14.75}{$R_i$}
  \foreach \x in {1.55,4.05,7.25,9.2,12.25} \LFiveRelease{\x}
  \LFiveDue{11.7}\LFiveDue{14.75}
  \draw[lf5/brace] (1,1.25)--(1.55,1.25) node[midway,above=3pt] {$1$};
  \draw[lf5/brace] (1.55,1.25)--(4.05,1.25) node[midway,above=3pt] {$P$};
  \draw[lf5/brace] (11.7,1.25)--(12.25,1.25) node[midway,above=3pt] {$1$};
  \draw[line width=.45pt] (11.7,.35)--(11.7,1.05);
  \node[anchor=east] at (11.43,1.28) {$d_{Z_i}$};
  \node[anchor=east] at (16.25,1.28) {$C_{R_i}=d_{R_i}$};
  \draw[lf5/axis] (.55,0)--(15.45,0);
  \node[anchor=west] at (15.48,0) {$t-b_i$};
  \LFiveTick{1}{-.15}{$0$}
  \LFiveTick{1.55}{-.66}{$1$}
  \LFiveTick{4.05}{-.15}{$P+1$}
  \LFiveTick{6.55}{-.15}{$2P+1$}
  \LFiveTick{7.25}{-.66}{$(m-1)P+1$}
  \LFiveTick{9.2}{-.15}{$mP$}
  \LFiveTick{9.75}{-.66}{$mP+1$}
  \LFiveTick{11.7}{-.15}{$(m+1)P$}
  \LFiveTick{12.25}{-.66}{$(m+1)P+1$}
  \LFiveTick{14.75}{-.15}{$(m+2)P+1$}
  \foreach \x in {1.55,7.25,9.75,12.25}
    \draw[black!45,line width=.35pt] (\x,-.07)--(\x,-.62);
  \node[anchor=west] at (0,-1.16)
    {Idle: $[b_i,b_i+1)$, length $1$. \quad $C_{Z_i}=d_{Z_i}+1=r_{R_i}$; guard unchanged.};
  \node[anchor=west] at (0,-1.57)
    {Tail members from this vertex: $A_{g_1},\ldots,A_{g_m}$, with weights $U+\ell_{e(g_k)}$.};
\end{scope}

\draw[black!18] (0,-5.99)--(16.4,-5.99);
\node[anchor=west,font=\small\bfseries] at (0,-6.40)
  {(c) Global location of the $A/B$ due dates and the tail};
\node[anchor=west] at (0,-6.89)
  {$d_{A_g}=d_{B_g}=d_g=H-(g-1)L$, \quad $P=ML$, \quad $g=1,\ldots,M$.};

% Global timing is a separate scale. All incidence due dates are in the
% last guard interval, strictly after H-P and no later than H.
\begin{scope}[yshift=-9.25cm]
  \node[anchor=west,align=left] at (.05,.72)
    {Every early $A/B$-job\\finishes by $H-2P$};
  \draw[-{Stealth[length=1.2mm]},line width=.45pt] (2.87,.64)--(3.20,.35);
  \node at (4.65,.70) {$\cdots$};
  \LFiveGuardJob{6.2}{9.2}{$R_n$}
  \draw[lf5/base,fill=black!5] (9.2,.35) rectangle (12.2,1.05);
  \node[align=center] at (10.7,.70) {first tail job};
  \draw[lf5/base,fill=black!5] (12.2,.35) rectangle (15.2,1.05);
  \node at (13.7,.70) {$\cdots$};
  \foreach \x in {3.2,6.2,9.2,12.2,15.2}
    \draw[lf5/guide] (\x,0)--(\x,.35);
  \draw[lf5/axis] (2.85,0)--(15.65,0);
  \node[anchor=west] at (15.72,0) {$t$};
  \LFiveTick{3.2}{-.15}{$H-2P$}
  \LFiveTick{6.2}{-.15}{$H-P$}
  \LFiveTick{9.2}{-.15}{$H$}
  \LFiveTick{12.2}{-.15}{$H+P$}
  \LFiveTick{15.2}{-.15}{$H+2P$}
  \LFiveRelease{6.2}
  % Three representative due-date markers; dots suppress all intermediate ranks.
  \draw[black!45,line width=.4pt] (6.8,1.05)--(6.8,1.46);
  \draw[black!45,line width=.4pt] (8.6,1.05)--(8.6,1.46);
  \draw[black!45,line width=.4pt] (9.2,1.05)--(9.2,1.46);
  \foreach \x in {6.8,8.6,9.2}\node[lf5/due] at (\x,1.46) {};
  \node[anchor=south] at (6.8,1.56) {$d_M$};
  \node at (7.7,1.48) {$\cdots$};
  \node[anchor=south east] at (8.6,1.56) {$d_2$};
  \node[anchor=south west] at (9.2,1.56) {$d_1=H$};
  \draw[lf5/brace] (6.2,1.22)--(6.8,1.22) node[midway,above=3pt] {$L$};
  \draw[lf5/brace] (9.2,1.27)--(12.2,1.27)
    node[midway,above=3pt] {$P$};
  \node[anchor=west,align=left] at (.05,-.86)
    {$H-P<d_M\leq d_g\leq H$;\quad every early $A/B$-job is on time.};
  \node[anchor=west] at (.05,-1.34)
    {Every tail job starts at or after $H$, so its completion time $C_g\geq H+P>d_g$ and its tardiness $T_g\geq P$.};
\end{scope}

\draw[black!18] (0,-10.95)--(16.4,-10.95);
\node[anchor=west,font=\small\bfseries] at (0,-11.34) {Exact release and due dates};
\node[anchor=north west,inner sep=0pt] at (.05,-11.65) {%
  \renewcommand{\arraystretch}{1.22}%
  \setlength{\tabcolsep}{7pt}%
  \begin{tabular*}{16.15cm}{@{\extracolsep{\fill}}llll@{}}
    Job & Release date & Due date & Weight\\\hline
    $A_{g_k}$ & $b_i+(k-1)P$ & $H-(g_k-1)L$ & $U+\ell_{e(g_k)}$\\
    $B_{g_k}$ & $b_i+(k-1)P+1$ & $H-(g_k-1)L$ & $U+h_{e(g_k)}$\\
    $Z_i$ & $b_i+mP$ & $b_i+(m+1)P$ & $F_i$\\
    $R_i$ & $b_i+(m+1)P+1$ & $b_{i+1}=b_i+(m+2)P+1$ & $\overline K+1$
  \end{tabular*}};
\end{tikzpicture}
\par\vspace{2pt}
\parbox{\linewidth}{\footnotesize\textit{Scale convention.}
All processing rectangles in (a) and (b) have the same width, representing $P$;
the one-unit offsets are enlarged for visibility and the ellipses omit middle
incidences. Panel (c) has a separate scale; the $L$ spacing is schematic.
All displayed time labels and formulas are exact.}
\caption{The two normalized vertex modes in Lemma~\ref{lem:consistency}.
In the $A$-mode the charge is on time and the one-unit idle interval precedes
the guard. In the $B$-mode that idle interval is at the start of the block,
all early jobs and the charge start one unit later, and the charge costs $F_i$.
The guard has identical start and completion times in both modes, so no delay
passes to the next block. Complementary jobs remain in the tail; they need
not be consecutive there. Panel (c) shows why all early jobs are on time and
all tail jobs are tardy. The tail is shown compacted from $H$.}
\label{fig:lemma5-modes}
\endgroup
\end{figure}
\afterpage{\clearpage}

\begin{proof}
There are $M$ early $A/B$-jobs. Each of the $n$ windows can fit at most $m$ jobs by Lemma~\ref{lem:controls} and the window-capacity paragraph following it. Since $M=nm$, the $M$ early jobs saturate all $n$ windows, so every window contains exactly $m$ jobs. Let $a_g=b_i+(k-1)P$ for $g=(i-1)m+k$, with $1\le i\le n$ and $1\le k\le m$. By the charge-start bound in Lemma~\ref{lem:controls}, the job in local position $k$ and its $m-k$ successors must finish by $b_i+mP+1$. Its start $s$ therefore satisfies $s\le b_i+mP+1-(m-k+1)P=a_g+1$. By \eqref{eq:blocks} and Table~\ref{tab:jobs}, consecutive $a_g$ values differ by $P$ within a block and by $3P+1$ between blocks. Thus $a_{g+1}-a_g\ge P>1$, using \eqref{eq:length}. Hence a job of rank $h>g$ is released at or after $a_h\ge a_g+P>a_g+1$ and cannot occupy this position. If $\rho_g$ is the incidence rank of the early $A/B$-job in chronological position $g$, it follows that $\rho_g\le g$. Exactly one early member of every job pair is present, so $(\rho_1,\ldots,\rho_M)$ is a permutation of $1,\ldots,M$. Therefore $\sum_{g=1}^{M}(g-\rho_g)=0$, while every summand is nonnegative, giving $\rho_g=g$ for every $g$. Each job pair contributes its early member in its own window and position; this conclusion did not assume that the initial early schedule had no idle time.

Left shift the resulting permutation using \eqref{eq:leftshift}. Cost cannot increase; each guard remains at its release by Lemma~\ref{lem:controls}. The order relative to the guards is unchanged, so the early $A/B$-jobs remain in their own windows. For $i>1$, the job preceding the first such job of vertex $i$ is guard $R_{i-1}$, which completes at $b_i$; for $i=1$, the machine is initially available at $b_1=0$. Thus no delay is inherited from the preceding block. Write a local earliest start as $b_i+(k-1)P+\varepsilon_k$. The recurrence is
\[
\varepsilon_0=0,\qquad
\varepsilon_k=\max\{\varepsilon_{k-1},z_k\}\quad(1\le k\le m),
\]
where $z_k=0$ for a selected $A$ and $z_k=1$ for a selected $B$. Thus $\varepsilon_k=\max_{1\le h\le k}z_h\in\{0,1\}$. The final early completion is $b_i+mP+\varepsilon_m$, so
\[
s_{Z_i}=b_i+mP+\varepsilon_m,\qquad
T_{Z_i}=\varepsilon_m.
\]
% A8.4.4: consolidate repeated maximum-offset consequences
Thus the charge contribution is $F_i\varepsilon_m$: it is zero when all selected jobs are $A$-jobs and equals $F_i$ as soon as a $B$-job is selected.
The guard following $Z_i$ starts at
\[
\max\{b_i+(m+1)P+1,\ b_i+(m+1)P+\varepsilon_m\}
=b_i+(m+1)P+1,
\]
since $\varepsilon_m\le1$. It completes at $b_{i+1}$, which also verifies the block-boundary reset directly. The recurrence takes a maximum of the release offsets, so repeated $B$ selections do not add further one-unit delays.

 Thus the recurrence gives the two mode timings: in the $A$-mode,
$\varepsilon_k=0$ for every $k$, so one unit of idle time remains immediately
before the guard; in the $B$-mode, $\varepsilon_k=1$ for every $k$, so the
charge job finishes exactly when the guard starts. 

 If a vertex has both $A$-jobs and $B$-jobs early, let $\mathcal A_i$ be
the nonempty set of its incidence ranks whose $A_g$ is early. %If a vertex uses both kinds, let $\mathcal A_i$ be the nonempty set of its incidence ranks whose $A_g$ is early. 
 For each $g\in\mathcal A_i$, move $B_g$ from the tail into the early
position of $A_g$, and move $A_g$ to the old tail slot of $B_g$. Schedule
the resulting early $B$-job at vertex $i$ at the start
$b_i+(k-1)P+1$, where $g=(i-1)m+k$. Charge job $Z_i$ already contributed $F_i$ to total weighted tardiness and continues to do so. No other window or guard is affected. In each corresponding tail slot, replace $B_g$ by $A_g$, retaining its completion time $C_g$. Both jobs have due date $d_g=H-(g-1)L$, both are released before $H$, and $C_g\ge H+P>d_g$. The exact reduction in total cost is therefore
\begin{equation}
\sum_{g\in\mathcal A_i}
\bigl[(U+h_{e(g)})-(U+\ell_{e(g)})\bigr](C_g-d_g)
=\sum_{g\in\mathcal A_i}\Delta_{e(g)}(C_g-d_g)
\ge P\sum_{g\in\mathcal A_i}\Delta_{e(g)}>0.
\label{eq:mixedsaving}
\end{equation}
All charge-job contributions remain unchanged, and every replacement is within one job pair $(A_g,B_g)$ with a common due date. Thus this strict saving follows directly from positive coefficients and positive tail tardiness; the due-date-gap margin used in Lemma~\ref{lem:exchange} is unnecessary here. Apply this operation once to every vertex using both job types. The resulting schedule has the claimed modes, with all charge and guard placements and all releases respected. The transformation is polynomial and does not require optimality.
\end{proof}

Call a schedule \emph{canonical} if every vertex uses one of the two modes of Lemma~\ref{lem:consistency}, its charge starts immediately after its early $A/B$-jobs, every guard starts at its release, and all remaining $A/B$-jobs run consecutively from $H$ in nonincreasing weight order.  The early job at local position $k$ of vertex $i$ starts at $b_i+(k-1)P$ in the $A$-mode, or at $b_i+(k-1)P+1$ in the $B$-mode, for $1\le i\le n$ and $1\le k\le m$. Jobs in the tail with  equal weights may be ordered arbitrarily.

\begin{proposition}[Normalization of arbitrary bounded-cost schedules]\label{prop:normalization}
Every feasible schedule of cost at most $\UB$, including one with noninteger starts, admits a canonical schedule of no greater cost. Given its job permutation, such a schedule can be constructed in polynomial time.
\end{proposition}
\begin{proof}
Apply Lemma~\ref{lem:exchange}.   This first left-shifts the input schedule,
if necessary, and then produces an integral schedule of no greater cost with
exactly one early member of each pair $(A_g,B_g)$. Next apply
Lemma~\ref{lem:consistency}, which converts the schedule, again without
increasing cost, so that every vertex uses one of the two modes and the
charge and guard placements have the stated form. 

It remains to arrange the tail. Compact the tail from $H$ and sort it in
nonincreasing weight order. Every tail job is already released and has due
date at most $H$, so it remains tardy in every tail position. For the fixed
tail set $\mathcal T$, its objective is
\[
\sum_{j\in\mathcal T}w_jC_j-\sum_{j\in\mathcal T}w_jd_j.
\]
The second term is independent of order. If adjacent jobs $a,b$ have
$w_a<w_b$, swapping them changes the weighted completion part by
$P(w_a-w_b)<0$. Ties change it by zero. Consequently nonincreasing weights
minimize the tail cost; this is the equal-processing-time specialization of
Smith's rule \citep{smith1956}. Sorting changes neither feasibility nor any
early contribution.

There are at most $M$ exchanges for missing early members of the job pairs,
at most $n$ vertex-mode conversions, and one tail sort. All computations use
only the input permutation and polynomial-bit data from the reduction; the
possibly noninteger original start times need not be encoded or manipulated.
\end{proof}

\section{Cost identity and strong NP-completeness}\label{sec:identity}
% A9.4 / N2 and A8.4.4: define C(x), consolidate tail-order proof

Throughout this section, schedules are canonical in the sense of
Proposition~\ref{prop:normalization}. For each vertex $i\in\mathcal V$,
set $x_i=1$ if vertex $i$ uses the $A$-mode, so that all its $A$-jobs are
early and the corresponding $B$-jobs are in the tail. Set $x_i=0$ if
vertex $i$ uses the $B$-mode, so that all its $B$-jobs are early and the
corresponding $A$-jobs are in the tail. Write
$x=(x_i)_{i\in\mathcal V}$. Conversely, every vector
$x\in\{0,1\}^n$ defines a feasible canonical schedule: use the mode timings
established in Lemma~\ref{lem:consistency}, then process the remaining
$A/B$-jobs consecutively from $H$ in nonincreasing weight order. All tail
jobs are released before $H$ by \eqref{eq:abrelease}. For the fixed tail
set, Proposition~\ref{prop:normalization} shows that this order minimizes
the tail contribution, with equal-weight ties irrelevant.

By Lemma~\ref{lem:consistency}, the charge-job contribution at vertex $i$
is $F_i(1-x_i)$. Early $A/B$-jobs and guards contribute zero by
Lemma~\ref{lem:controls} and the early/tail separation established after
its proof. Let $C(x)$ denote the total weighted tardiness of the canonical
schedule associated with $x$; the order of equal-weight tail jobs does not
affect this value.

\begin{proposition}[Cost identity for positive edge coefficients]\label{prop:coefficients}
Let $n\ge2$, and assign to every edge $e\in\mathcal E$ of the complete graph $(\mathcal V,\mathcal E)$ a positive integer coefficient $\Delta_e$. Using the coefficient array $(\Delta_e)_{e\in\mathcal E}$, define the job construction by \eqref{eq:intervals}--\eqref{eq:blocks} and Table~\ref{tab:jobs}, without imposing the graph-specific choice \eqref{eq:graphcoefficients} or a source threshold. Recompute all derived parameters from this coefficient array, including $U,G,Q,P,L,F_i$, and $\UB$. The normalization proposition remains valid, and the resulting scheduling instance has optimum
\[
\OPT=\UB-\frac P2\max_{x\in\{0,1\}^n}
\sum_{1\le a<b\le n}\Delta_{\{a,b\}}|x_a-x_b|.
\]
\end{proposition}
\begin{proof}
 Equations \eqref{eq:intervals}--\eqref{eq:V} give $1\le\ell_e<h_e\le V$ and a one-unit gap between successive intervals. The definitions $q_e=M-2t_e$ and $M=2E$ give $q_e=2(E-t_e)\ge0$. Also $P_0=2M$ is even, $g_{i,e}-1\ge0$, and each of the $n-1$ incident coefficients is a positive integer and hence at least one. Thus \eqref{eq:fees} gives
\[
f_i\ge \frac32 P_0\sum_{e\in\mathcal E_i}\Delta_e
\ge \frac32 P_0(n-1)>0.
\]
% A9.2 / M1,R1: correct parameter ranges and reference
Thus $f_{\min}>0$, so $Q$ is well defined. The quantities $U,G,Q,P,L,F_i$, and $\UB$ are positive integers. The quantities $q_e$, block boundaries, and release dates are nonnegative integers in the ranges already specified. With $\delta=2P_0f_{\min}Q-G$, the identities established in Lemma~\ref{lem:parameter-bounds} give the following strict margins.
\begin{align*}
P&=ML>1,\\
\left(\min_{1\le i\le n}F_i\right)(2P+1)-\UB&=Q(\delta+f_{\min})>0,\\
UL-V(2M+1)P-\FS&=Q(U+1)>0,\\
PU-\FS&=Q\left(P_0U-\sum_{i=1}^{n} f_i\right)>0.
\end{align*}
Here $1\le\delta\le2P_0f_{\min}$ follows from the floor definition of the recomputed $Q$.  This is \eqref{eq:Q}; see the calculation of $\delta$ in the proof of Lemma~\ref{lem:parameter-bounds}. The charge and exchange equalities above are \eqref{eq:charge-margin} and \eqref{eq:exchange-margin}, reapplied to the recomputed parameters. No graph-specific coefficient bound occurs in these identities.

The release and due-date identities \eqref{eq:abrelease} and \eqref{eq:duebound} depend only on $P=ML$ and the block definitions. The guard weight remains $\UB+1$. Thus Lemma~\ref{lem:controls} applies, compaction gives \eqref{eq:tailbound},  and Lemmas~\ref{lem:matching}, \ref{lem:exchange}, and \ref{lem:consistency} remain valid: they use positive interval endpoints, the release/due-date identities, and inequalities \eqref{eq:length}--\eqref{eq:addbound}, all re-established above for the new coefficients. In particular, Lemma~\ref{lem:consistency} uses only the strict positivity of each $\Delta_e$. The reference schedule in which every vertex uses the $A$-mode still has cost $\UB$, by its construction and evaluation in Section~\ref{sec:construction-bounds}, using \eqref{eq:G} and \eqref{eq:actual}. An optimum exists by the finite earliest-permutation argument following \eqref{eq:leftshift}; the reference schedule places its cost within the normalization bound. Proposition~\ref{prop:normalization} therefore applies to the supplied coefficient array.

Fix $e=\{a,b\}\in\mathcal E$ with $a<b$ and abbreviate $g_a=g_{a,e}$, $g_b=g_{b,e}$, $q=q_e$, $\Delta=\Delta_e$, and $u=U+\ell_e$. The two tail weights are $u+\Delta x_a$ and $u+\Delta x_b$. This follows from the mode definition at the start of this section, the weights in Table~\ref{tab:jobs}, and $h_e=\ell_e+\Delta_e$ in \eqref{eq:intervals}.  By the canonical weight order in Proposition~\ref{prop:normalization} and the disjoint ordered intervals in \eqref{eq:intervals}, the $E-t_e$ higher intervals precede these two jobs. Each higher interval contains exactly two tail jobs by Lemma~\ref{lem:consistency}; hence precisely $2(E-t_e)=q_e=q$ tail jobs precede them, independently of the bits. The two completions are $H+(q+1)P$ and $H+(q+2)P$. Table~\ref{tab:jobs} gives $d_{g_a}=H-(g_a-1)L$ and $d_{g_b}=H-(g_b-1)L$, and the larger weight goes first by Proposition~\ref{prop:normalization}. Their combined tail cost is
\begin{equation}
\begin{aligned}
&P\bigl[(q+1)\max\{u+\Delta x_a,u+\Delta x_b\}\\
&\hspace{12mm}+(q+2)\min\{u+\Delta x_a,u+\Delta x_b\}\bigr]\\
&\quad+L\bigl[(g_a-1)(u+\Delta x_a)+(g_b-1)(u+\Delta x_b)\bigr].
\end{aligned}
\label{eq:edge-tail}
\end{equation}
By \eqref{eq:fees} and the scaling definitions \eqref{eq:actual},
\[
F_v=\sum_{e\in\mathcal E_v}\Delta_e
\left[P\left(q_e+\frac32\right)+L(g_{v,e}-1)\right].
\]
Lemma~\ref{lem:consistency} multiplies each endpoint's summand by $1-x_v$. We therefore allocate to edge $e$ the following portions of the charge contributions of $Z_a$ and $Z_b$:
\begin{equation}
\begin{aligned}
&\Delta\left[P\left(q+\frac32\right)+L(g_a-1)\right](1-x_a)\\
&\quad+\Delta\left[P\left(q+\frac32\right)+L(g_b-1)\right](1-x_b).
\end{aligned}
\label{eq:edge-charge}
\end{equation}
 Define
\begin{equation}
\begin{aligned}
W_e&=P(2q+3)+L(g_a+g_b-2),\\
D_e&=(u+\Delta)W_e\\
&=(U+h_e)\bigl[P(2q_e+3)+L(g_{a,e}+g_{b,e}-2)\bigr],
\end{aligned}
\label{eq:edge-baseline}
\end{equation}
and define $\Phi_e(x_a,x_b)$ as the sum of the tail contribution \eqref{eq:edge-tail} and the allocated charge contribution \eqref{eq:edge-charge}.
To expose the cancellation, put $\sigma_e=x_a+x_b$ and $\eta_e=|x_a-x_b|$. Since the modes are binary,
\[
\max\{x_a,x_b\}=\frac{\sigma_e+\eta_e}{2},\qquad
\min\{x_a,x_b\}=\frac{\sigma_e-\eta_e}{2}.
\]
Substituting these binary max/min identities into \eqref{eq:edge-tail}, and using \eqref{eq:edge-baseline}, gives the tail contribution
\[
uW_e+\frac{P\Delta}{2}\bigl[(2q+3)\sigma_e-\eta_e\bigr]
 +L\Delta\bigl[(g_a-1)x_a+(g_b-1)x_b\bigr].
\]
 Expanding \eqref{eq:edge-charge} with the same definition of $W_e$ gives the allocated charge contribution
\[
\Delta W_e-\frac{P\Delta}{2}(2q+3)\sigma_e
 -L\Delta\bigl[(g_a-1)x_a+(g_b-1)x_b\bigr].
\]
Thus the linear mode terms, including their incidence-specific due-date offsets, cancel exactly, and
 \begin{equation}
\begin{aligned}
\Phi_e(x_a,x_b)&=(u+\Delta)W_e-\frac{P\Delta}{2}\eta_e\\
&=D_e-\frac{P\Delta_e}{2}|x_a-x_b|.
\end{aligned}
\label{eq:edge-cost}
\end{equation}
The charge coefficients therefore leave only a constant and the saving for unequal endpoint modes.

% A8.4.5: original four-case expansions are proposed deletions, not removed.

In the equal-weight cases, either order gives the same combined cost even when $g_a\ne g_b$, because the weighted due-date sum is unchanged. As in the fixed-tail interchange argument in Proposition~\ref{prop:normalization}, when $w_a=w_b$, both $w_aC_a+w_bC_b$ and $w_ad_a+w_bd_b$ are unchanged by interchanging the jobs. Table~\ref{tab:cases}   records the resulting four values.

\begin{table}[htbp]
\centering
\caption{Edge contribution after normalization; values follow from
\eqref{eq:edge-cost}.}\label{tab:cases}
\setstretch{1.15}
\begin{tabular}{ccc}
\toprule
$x_a$ & $x_b$ & Edge contribution\\
\midrule
0 & 0 & $D_e$\\
0 & 1 & $D_e-P\Delta_e/2$\\
1 & 0 & $D_e-P\Delta_e/2$\\
1 & 1 & $D_e$\\
\bottomrule
\end{tabular}
\end{table}

 Every charge-job contribution has been allocated exactly once: for each vertex $v$, summing its endpoint allocations over $e\in\mathcal E_v$ gives $Qf_v(1-x_v)=F_v(1-x_v)$.  Indeed, \eqref{eq:edge-charge} includes one summand for each incident edge; summing these at $v$ reproduces \eqref{eq:fees} scaled by \eqref{eq:actual}. Early $A/B$-jobs and guards have zero cost, and $\sum_{e\in\mathcal E}D_e=QG=\UB$  by \eqref{eq:edge-baseline}, \eqref{eq:G}, and \eqref{eq:actual}; the factor $Q$ enters through both $P=QP_0$ and $L=QL_0$.  Summing \eqref{eq:edge-cost} over all $e\in\mathcal E$ now gives the complete schedule cost: Lemma~\ref{lem:consistency} and the early/tail separation leave no additional nonzero contributions.
\begin{equation}
C(x)=\UB-\frac P2\sum_{1\le a<b\le n}\Delta_{\{a,b\}}|x_a-x_b|.
\label{eq:cutcost}
\end{equation}
 Every bit vector defines a feasible canonical schedule by the mode-and-tail construction at the start of Section~\ref{sec:identity}. An optimum has cost at most $\UB$ by the reference schedule in Section~\ref{sec:construction-bounds}, and hence can be normalized by Proposition~\ref{prop:normalization}. Thus the optimum is the minimum of \eqref{eq:cutcost} over
$x\in\{0,1\}^n$, which is
\[
\UB-\frac P2\max_{x\in\{0,1\}^n}
\sum_{1\le a<b\le n}\Delta_{\{a,b\}}|x_a-x_b|.
\]
This proves the proposition.
\end{proof}

 Proposition~\ref{prop:coefficients} requires a strictly positive integer
coefficient on every edge of the auxiliary complete graph. It does not
apply, as stated, to zero, negative, or arbitrary real coefficients. Each
coefficient $\Delta_e$ is a construction parameter: by \eqref{eq:intervals}
it determines the interval width $h_e-\ell_e$; by Table~\ref{tab:jobs} it
determines the weight difference between the $A$- and $B$-jobs for both
endpoint incidences; and by \eqref{eq:fees} and \eqref{eq:actual} it
determines the corresponding portions of the endpoint charge weights.

 We now return to the specific coefficient choice in \eqref{eq:graphcoefficients}:
original graph edges have coefficient $\Gamma+1$, while added auxiliary
edges have coefficient $1$. 

\begin{theorem}\label{thm:main}
The decision version of $\problem$ is strongly NP-complete, even under $d_j\ge r_j+p$ for every job. In the reduction, all numerical data, including the decision threshold, are bounded by a polynomial in the number of jobs.
\end{theorem}
\begin{proof}
\medskip\noindent\textit{\underline{Correctness of the reduction.}}\quad
Let $c_{\mathcal G}(x)=\sum_{\substack{\{a,b\}\in\mathcal F\\a<b}}|x_a-x_b|$ be the number of original edges cut by $x$, and let $s(x)=\sum_{i=1}^{n}x_i$ be the size of one side.  Substituting \eqref{eq:graphcoefficients}, the coefficient $1$ assigned to
every auxiliary complete-graph edge contributes $s(x)(n-s(x))$, the number
of crossing pairs in the complete graph. The additional $\Gamma$ assigned to
original graph edges contributes $\Gamma c_{\mathcal G}(x)$. The upper bound
$s(x)(n-s(x))\le\lfloor n^2/4\rfloor=\Gamma-1$ follows from the definition
of $\Gamma$ in \eqref{eq:graphcoefficients}. Thus
\begin{equation}
\sum_{1\le a<b\le n}\Delta_{\{a,b\}}|x_a-x_b|
=\Gamma c_{\mathcal G}(x)+s(x)(n-s(x)),
\qquad 0\le s(x)(n-s(x))\le\Gamma-1.
\label{eq:graphscore}
\end{equation}
Combining \eqref{eq:threshold}, \eqref{eq:cutcost}, and \eqref{eq:graphscore} gives
\begin{equation}
C(x)-K_*=\frac P2\bigl[\Gamma(\kappa-c_{\mathcal G}(x))-s(x)(n-s(x))\bigr].
\label{eq:cutgap}
\end{equation}

For the forward direction, suppose the MAX-CUT instance is YES and choose $x$ with $c_{\mathcal G}(x)\ge\kappa$.  The mode timings of Lemma~\ref{lem:consistency}, followed by the tail construction at the start of Section~\ref{sec:identity}, give a feasible canonical schedule. Equation~\eqref{eq:cutgap} gives $C(x)\le K_*$, since $\kappa-c_{\mathcal G}(x)\le0$ and $s(x)(n-s(x))\ge0$.

For the reverse direction, suppose the constructed scheduling instance
has a feasible schedule with cost at most $K_*$. By
\eqref{eq:threshold},
\[
K_* = \UB-\frac{P}{2}\Gamma\kappa \le \UB,
\]
since $P>0$, $\Gamma>0$, and $\kappa\ge1$.
Thus the schedule satisfies the cost bound required by
Proposition~\ref{prop:normalization}, which gives a canonical schedule
of no greater cost. If $x$ denotes its mode vector, then $C(x)\le K_*$.
If its cut had $c_{\mathcal G}(x)\le\kappa-1$, then \eqref{eq:cutgap}, together with $s(x)(n-s(x))\le\Gamma-1$ from \eqref{eq:graphscore}, would imply
\[
C(x)-K_*\ge\frac P2\bigl[\Gamma-(\Gamma-1)\bigr]=\frac P2>0,
\]
a contradiction. Hence $c_{\mathcal G}(x)\ge\kappa$. Therefore, the MAX-CUT instance is YES if and only if the constructed scheduling instance admits a schedule of cost at most $K_*$. 

\medskip\noindent
\textit{\underline{Validity and polynomial bounds for the construction.}}\quad The threshold is integral because $P$ is even.  Indeed, $P=QP_0$ by \eqref{eq:actual} and $P_0=2M$; $\Gamma$ and $\kappa$ are integers as well, so \eqref{eq:threshold} is integral.  For nonnegativity, \eqref{eq:graphcoefficients} and \eqref{eq:V} give $V=2E+\Gamma|\mathcal F|\ge\Gamma\kappa$, since $\kappa\le|\mathcal F|$. Equation~\eqref{eq:U} gives $U>V$. Hence
\[
\Gamma\kappa\le\Gamma|\mathcal F|\le V<U.
\]
  In the reference schedule from Section~\ref{sec:construction-bounds}, every
vertex uses the $A$-mode. The $M$ jobs processed after $H$ have completion
offsets $tP$, for $1\le t\le M$; each of these jobs has weight at least $U$,
and each has due date at most $H$ by Table~\ref{tab:jobs}. %In the all-$A$ reference schedule of Section~\ref{sec:construction-bounds}, the $M$ tail completions have offsets $tP$, each tail weight is at least $U$, and each tail due date is at most $H$ by Table~\ref{tab:jobs}. 
Thus $\UB\ge UP\sum_{t=1}^{M}t=UPM(M+1)/2$. Since $M\ge2$, this exceeds
$P\Gamma\kappa/2$, so $K_*>0$.  Equations \eqref{eq:abrelease} and
\eqref{eq:duebound}, together with Table~\ref{tab:jobs}, show that every
constructed job satisfies $d_j\ge r_j+P$.
%Thus $\UB\ge UP\sum_{t=1}^{M}t=UPM(M+1)/2$. Since $M\ge2$, this exceeds $P\Gamma\kappa/2$, so $K_*>0$ on nontrivial source instances. The fixed preprocessing instances in Section~\ref{par:preprocessing} also have valid nonnegative thresholds. Equations \eqref{eq:abrelease} and \eqref{eq:duebound}, together with Table~\ref{tab:jobs}, show that every constructed job satisfies $d_j\ge r_j+P$; the fixed preprocessing instances described in Section~\ref{par:preprocessing} (one unit job for YES and two such jobs for NO, both with threshold zero) satisfy the same restriction.

The strict inequalities needed for schedule normalization were proved in Lemma~\ref{lem:parameter-bounds}; here only polynomial magnitude is needed.  By \eqref{eq:graphcoefficients} and \eqref{eq:V}, using $E=\binom n2$, $M=n(n-1)$ from Section~\ref{sec:construction-data}, and $\Gamma=O(n^2)$ from \eqref{eq:graphcoefficients},
\[
V=E+\sum_{e\in\mathcal E}\Delta_e
=2E+\Gamma|\mathcal F|
\le(2+\Gamma)\binom n2=O(n^4),\qquad M=O(n^2).
\]
Every $q_e$ and incidence rank is $O(M)$, $P_0=2M$, and $L_0=2$. Thus the coefficient of $\Delta_e$ in \eqref{eq:fees} is $O(M^2)$.  Each edge occurs twice in the incidence sum, and \eqref{eq:V} gives
\[
\sum_{i=1}^{n}\sum_{e\in\mathcal E_i}\Delta_e
=2\sum_{e\in\mathcal E}\Delta_e\le2V.
\]
Consequently \eqref{eq:fees} gives the first bound below, and \eqref{eq:U} gives the second:
\[
\sum_{i=1}^{n}f_i=O(VM^2),\qquad U=O(VM^2).
\]
Each bracket in \eqref{eq:G} is $O(M^2)$, $h_e\le V$ by \eqref{eq:intervals}--\eqref{eq:V}, and there are $E=M/2$ edges by Section~\ref{sec:construction-data}. Thus \eqref{eq:G} gives $G=O((U+V)M^3)$. Equation~\eqref{eq:Q} and positivity of the integers $P_0,f_{\min}$ give
\[
1\le Q=\left\lfloor\frac{G}{2P_0f_{\min}}\right\rfloor+1\le G+1.
\]
Table~\ref{tab:magnitudes} collects explicit degree estimates; the hidden constants are uniform and independent of the source graph.

\begin{table}[htbp]
\centering
\caption{Polynomial magnitude bounds in the number of retained graph vertices}\label{tab:magnitudes}
\setstretch{1.15}
\begin{tabular}{ll}
\toprule
Quantity & Magnitude bound in  $n$\\
\midrule
$\Gamma,\Delta_e,M,P_0,q_e,g_{i,e}$ & $O(n^2)$\\
$V$ & $O(n^4)$\\
$\sum_{i=1}^{n} f_i,\ U$ & $O(n^8)$\\
$G,\ Q,\ L$ & $O(n^{14})$\\
$P$ & $O(n^{16})$\\
$H,\ r_j,\ d_j$ & $O(n^{18})$\\
$F_i$ & $O(n^{22})$\\
$\UB,\ K_*,\ w_j$ & $O(n^{28})$\\
\bottomrule
\end{tabular}
\end{table}

 To verify the scaling rows of Table~\ref{tab:magnitudes}, the preceding estimate gives $G=O((U+V)M^3)=O(n^{14})$. Equation~\eqref{eq:Q} gives $Q\le G+1=O(n^{14})$. Since $P_0=2M$ and $L_0=2$, \eqref{eq:actual} gives $L=2Q=O(n^{14})$ and $P=2MQ=O(n^{16})$.  Next, \eqref{eq:blocks} gives $H=n((n+1)P+1)=O(n^{18})$, and Table~\ref{tab:jobs} bounds all release and due dates by $H$, giving the same exponent. Equation~\eqref{eq:actual} gives $F_i=Qf_i=O(n^{22})$ and $\UB=QG=O(n^{28})$. The threshold argument above gives $0\le K_*\le\UB$, so $K_*=O(n^{28})$.

 The largest weights are the guards' $\UB+1$ from Table~\ref{tab:jobs}. Indeed, \eqref{eq:U}, \eqref{eq:G}, and \eqref{eq:actual} give $F_i=Qf_i<QU<QG=\UB$. Also $U+h_e\le U+V<G\le QG$: for the last edge, $h_e=V$ and $q_e=0$, so its summand in \eqref{eq:G} alone exceeds $U+V$. Thus all nonguard weights are below $\UB+1$. Therefore, every target integer has $O(\log n)$ bits and magnitude $O(n^{28})$.  The job construction in Section~\ref{sec:construction-data} gives $N=2M+2n=2n^2$ jobs: two for each incidence and a charge and guard for each vertex. Hence $O(n^{28})=O(N^{14})$. There are $O(N)$ numerical job entries, each of unary length $O(N^{14})$, and one threshold, giving total unary encoding length $O(N^{15})$.  Since the source graph has no isolated vertices, $n\le2|\mathcal F|$, so this unary output length is polynomial in the graph input length.
%The preprocessing in Section~\ref{par:preprocessing} ensures $n\le2|\mathcal F|$, so this unary output length is polynomial in the graph input length. 
The reduction therefore proves strong NP-hardness.

All parameters can be computed with $O(n^2)$ arithmetic operations, using prefix sums for \eqref{eq:intervals} and direct sums for \eqref{eq:fees} and \eqref{eq:G}. Their operands have polynomial bit length, so standard integer arithmetic, including the division in \eqref{eq:Q}, takes polynomial time. The coefficients $\Delta_e$ specify interval endpoints, not numbers of jobs or time-indexed objects to enumerate.

\medskip\noindent\textit{\underline{Membership in NP.}}\quad
% A9.3 / M2: NP-membership proof for arbitrary inputs
For an arbitrary input instance with $N$ jobs and common processing time $p$, a job permutation is a certificate  using $O(N\log N)$ bits. Compute its earliest schedule by \eqref{eq:leftshift}. Writing $r_{\max}=\max_j r_j$ and $w_{\max}=\max_j w_j$,  induction in \eqref{eq:leftshift} gives $\widehat C_k\le r_{\max}+kp$, so every computed completion time is at most $r_{\max}+Np$, and total weighted tardiness is at most
\[
Nw_{\max}\bigl(r_{\max}+Np+\max_j|d_j|\bigr).
\]
The certificate and these computations have polynomial bit length. If any feasible schedule with the given permutation meets the decision threshold, its earliest schedule does too, including when original start times are noninteger,  by the dominance argument immediately following \eqref{eq:leftshift}. Thus the decision problem belongs to NP. Together with the reduction above, this proves strong NP-completeness.
\end{proof}

\medskip\noindent\textbf{Remark (Interpretation of the optimum).}\quad
Retain the notation from the proof of Theorem~\ref{thm:main}:
$c_{\mathcal G}(x)$ is the number of original graph edges cut by $x$,
and $s(x)=\sum_{i=1}^{n}x_i$ is the size of one side of the cut.
Proposition~\ref{prop:coefficients}, specialized using
\eqref{eq:graphscore}, gives
\begin{equation}
\OPT=\UB-\frac P2\max_{x\in\{0,1\}^n}
\bigl[\Gamma c_{\mathcal G}(x)+s(x)(n-s(x))\bigr].
\label{eq:optcut}
\end{equation}
Thus minimizing the scheduling objective first maximizes the number
of original graph edges cut and, among maximum cuts, maximizes
$s(x)(n-s(x))$. Indeed, by \eqref{eq:graphscore},
$0\le s(x)(n-s(x))\le\Gamma-1$. Losing one original cut edge
therefore reduces the first term in the maximized expression by
$\Gamma$, while the largest possible increase in the second term
is only $\Gamma-1$.
\par

\section{A phase-grid assignment algorithm and its shifted-objective guarantee}
\label{sec:approximation}

This section complements the strong NP-hardness result with a polynomial-time
phase-grid assignment algorithm and a worst-case analysis. The analysis applies to arbitrary
instances of $\problem$ and is independent of the reduction in the preceding
sections.

\subsection{Shifted objective}

For a feasible schedule $S$, write
\[
\begin{aligned}
 F(S)&=\sum_jw_j(C_j-d_j)^+,\\
 W&=\sum_jw_j,\\
 F^*&=\min_S F(S).
\end{aligned}
\]
We analyze the schedule-independent shift
\begin{equation}
\begin{aligned}
 B(I)&=pW,\\
 \Phi(S)&=F(S)+B(I)=\sum_jw_j\bigl(T_j(S)+p\bigr).
\end{aligned}
 \label{eq:shift}
\end{equation}
Thus $\Phi^*=F^*+pW>0$, even when all jobs can be completed on time. The
shift scales naturally under changes of time or weight units, does not depend
on the algorithm's output, and is independent of the schedule. Hence $F$ and
$\Phi$ have the same ordering and the same minimizers. The shift does, however,
change approximation ratios, so the guarantee
below is explicitly a guarantee for $\Phi$, not for the unshifted objective
$F$.

\subsection{The phase-grid assignment algorithm}

For $a\in[0,p)$, the phase-$a$ grid consists of starts $a+kp$,
$k\in\mathbb Z$, subject to nonnegative starting times. Let
\[
 \mathcal A=\{r_j\bmod p:1\le j\le N\}.
\]
For each distinct phase $a\in\mathcal A$, the algorithm performs the following
steps.
\begin{enumerate}
\item Round every release date upward to the phase-$a$ grid:
\[
 R_j(a)=a+p\left\lceil\frac{r_j-a}{p}\right\rceil.
\]
\item Generate the candidate starts
\[
 \mathcal T_a=\{R_i(a)+kp:1\le i\le N,\ 0\le k<N\}.
\]
\item Construct a bipartite graph with one node for every job and one node
for every distinct $t\in\mathcal T_a$. Include edge $(j,t)$ if $t\ge r_j$,
with cost $w_j(t+p-d_j)^+$, and find a minimum-cost matching covering every
job.
\item Interpret the matching as a schedule and retain the least-cost schedule
over all phases.
\end{enumerate}
Different slots on one grid define disjoint intervals of length $p$, so every
matching is feasible. Unmatched slots represent idle time; the matching must
cover the jobs but need not cover all slots.

\begin{lemma}
\label{lem:grid-slots}
For a fixed phase $a$, the matching computes an optimal schedule among all
schedules whose starts lie on the phase-$a$ grid.
\end{lemma}
\begin{proof}
Take any feasible grid schedule and preserve its job order. Left justify it
within the grid: start the first job at its rounded release and each successor
at the maximum of its rounded release and the preceding completion time. This
cannot increase any completion time. Every maximal consecutive block begins
at some $R_i(a)$, and successive jobs in that block start at
$R_i(a)+kp$ for $0\le k<N$. Hence all starts belong to $\mathcal T_a$.
Conversely, any matching covering all jobs gives a feasible grid schedule
whose objective value is exactly its total edge cost.
\end{proof}

There are at most $N$ phases, at most $N^2$ slots per phase, and at most
$N^3$ assignment edges per phase. A rectangular Hungarian implementation
uses $O(N^2M)$ arithmetic operations for $N$ jobs and 
%{\color{red}$S_a = |{\cal T}_a|$} 
$M = |{\cal T}_a|$ slots, so the full
algorithm uses $O(N^5)$ arithmetic operations. All generated times and costs
have polynomial binary length. In particular, the algorithm does not
enumerate all integer times, every residue $0,\ldots,p-1$, or the full time
horizon.

\subsection{Worst-case analysis}

For any job order  $\pi$, left shifting produces its earliest feasible schedule:
\[
 s_{\pi(1)}=r_{\pi(1)},\qquad
 s_{\pi(k)}=\max\{r_{\pi(k)},s_{\pi(k-1)}+p\}\quad(k\ge2).
\]
It follows that an optimum exists in which every start has the form $r_i+kp$
for some $0\le k<N$.  Consequently, we can choose an earliest optimal schedule in which the
residue of every start belongs to $\mathcal A$.

\begin{lemma}
\label{lem:phase-rounding}
If all starts in a feasible schedule are rounded upward to the same phase-$a$
grid, the resulting schedule is feasible. If job $j$ is delayed by
$\delta_j(a)$, then $0\le\delta_j(a)<p$ and its cost increases by at most
$w_j\delta_j(a)$.
\end{lemma}
\begin{proof}  For each job \(j\), let \[ \widehat{s}_j =a+p\left\lceil\frac{s_j-a}{p}\right\rceil \] denote its start time after upward rounding to the phase-\(a\) grid, and let \[ \delta_j(a)=\widehat{s}_j-s_j. \] 
Since \(x\leq\lceil x\rceil<x+1\), we have $0\leq\delta_j(a)<p.$ 
Furthermore,  \(\widehat{s}_j\geq s_j\geq r_j\), so all release-date constraints remain satisfied. Consider two successive jobs \(j\) and \(k\) in the original schedule. Feasibility of that schedule implies \(s_k\geq s_j+p\). Therefore, \[ \begin{aligned} \widehat{s}_k &=a+p\left\lceil\frac{s_k-a}{p}\right\rceil\\ &\geq a+p\left\lceil\frac{s_j+p-a}{p}\right\rceil\\ &=a+p\left\lceil\frac{s_j-a}{p}\right\rceil+p\\ &=\widehat{s}_j+p. \end{aligned} \] Thus, the rounded processing intervals remain disjoint, and the resulting schedule is feasible. The rounded completion time of job \(j\) is \(\widehat{C}_j=C_j+\delta_j(a)\). %Using \[ (x+\delta)^+\leq x^++\delta \qquad\text{for }\delta\geq0, \] we obtain \[ w_j\!\left[ \bigl(\widehat{C}_j-d_j\bigr)^+ -\bigl(C_j-d_j\bigr)^+ \right] \leq w_j\delta_j(a). \] 
Hence the weighted-tardiness cost of job \(j\) increases by at most \(w_j\delta_j(a)\).  \end{proof}
\begin{comment}
\begin{proof}
All release dates remain satisfied. If two successive original starts obey
$s_k\ge s_j+p$, monotonicity of upward rounding and its invariance under
translation by $p$ give {\color{red} $R_a(s_k) = a + p\lceil(t-a)/p\rceil \ge R_a(s_j)+p$},{\color{brown}***this is not clear enough. also it uses the same symbol R already used earlier**} so the rounded intervals
remain disjoint. The cost bound follows from
$(x+\delta)^+\le x^++\delta$ for $\delta\ge0$.
\end{proof}
\end{comment}

Uniformly averaging the grid phase over $[0,p)$ gives mean displacement $p/2$
for each fixed start and hence the additive bound $pW/2$. Weight-dependent
averaging over the phases actually enumerated by the algorithm yields the
sharper result below.

\begin{theorem}
\label{thm:shifted-approx}
Let $S_{\mathrm{alg}}$ be the schedule returned by the phase-grid assignment
algorithm, and define
\[
 F_{\mathrm{alg}}=F(S_{\mathrm{alg}}),\qquad
 \Phi_{\mathrm{alg}}=\Phi(S_{\mathrm{alg}})=F_{\mathrm{alg}}+pW.
\]
Then
\begin{equation}
 F_{\mathrm{alg}}\le F^*+D(I),\qquad
 D(I)=\frac{p}{2W}\left(W^2-\sum_jw_j^2\right).
 \label{eq:strongadd}
\end{equation}
Consequently,
\begin{equation}
 \frac{\Phi_{\mathrm{alg}}}{\Phi^*}
 \le \frac32-\frac{\sum_jw_j^2}{2W^2}
 \le \frac32-\frac1{2N}<\frac32.
 \label{eq:shifted-ratio}
\end{equation}
\end{theorem}
\begin{proof}
Fix an earliest optimal schedule and let $a_i=s_i^*\bmod p$. Every $a_i$ is
tested. For each $i$, round all optimal starts upward to the phase-$a_i$ grid,
and let $\delta_j(a_i)$ be the resulting delay of job $j$. By
Lemmas~\ref{lem:grid-slots} and~\ref{lem:phase-rounding},
\[
 F_{\mathrm{alg}}\le F^*+\min_i\sum_jw_j\delta_j(a_i).
\]
Choose index $i$ with probability $w_i/W$ for the analysis only. Then
\[
 \min_i\sum_jw_j\delta_j(a_i)
 \le\frac1W\sum_i\sum_jw_iw_j\delta_j(a_i).
\]
When $a_i=a_j$, both cross-displacements are zero. Otherwise, the two forward
distances around a circle of circumference $p$ satisfy
\[
 \delta_j(a_i)+\delta_i(a_j)=p.
\]
The diagonal terms vanish. Grouping the double sum into unordered pairs gives
\[
 \frac1W\sum_i\sum_jw_iw_j\delta_j(a_i)
 =\frac pW\sum_{\substack{i<j\\a_i\ne a_j}}w_iw_j
 \le\frac pW\sum_{i<j}w_iw_j
 =D(I),
\]
which proves~\eqref{eq:strongadd}. Adding the shift to both the algorithmic and
optimal values and using $\Phi^*\ge pW$ yields
\[
 \frac{\Phi_{\mathrm{alg}}}{\Phi^*}
 \le1+\frac{D(I)}{F^*+pW}
 \le1+\frac{D(I)}{pW}
 =\frac32-\frac{\sum_jw_j^2}{2W^2}.
\]
Finally, $\sum_jw_j^2\ge W^2/N$ by Cauchy--Schwarz. The algorithm itself is
deterministic; the random choice is used only in the proof.
\end{proof}

For $N=1$, $D(I)=0$ and the algorithm is exact. More generally, let $q^*$ be
the number of distinct phases among the optimal starts used in the proof, and
let $q=|\mathcal A|$; then $q^*\le q$. If $W_h$ is the total weight of the
jobs in optimal-start phase $h$, grouping the pair terms by phase and applying
Cauchy--Schwarz gives
\[
 F_{\mathrm{alg}}-F^*
 \le \frac{p}{2W}\left(W^2-\sum_{h=1}^{q^*}W_h^2\right)
 \le \frac{pW}{2}\left(1-\frac1{q^*}\right)
 \le \frac{pW}{2}\left(1-\frac1q\right).
\]
Hence the ratio is also at most $3/2-1/(2q)$. When all releases share one
phase, the algorithm is exact.

\subsection{Tightness and the role of the shift}

\begin{proposition}
\label{prop:grid-tight}
For every integer $N\ge1$, an integer-data instance attains ratio
$3/2-1/(2N)$ for the phase-grid assignment algorithm and the
shift~\eqref{eq:shift}.
\end{proposition}
\begin{proof}
Set $p=N$. For $i=0,\ldots,N-1$, create a unit-weight job with
\[
 r_i=i(2N+1),\qquad d_i=r_i+N.
\]
Executing each job at its release is feasible and has zero cost. Thus $F^*=0$,
$W=N$, and $\Phi^*=N^2$. The release residues are $0,1,\ldots,N-1$. At any
tested phase, rounding the releases produces a permutation of the delays
$0,1,\ldots,N-1$. The rounded intervals remain disjoint because consecutive
releases are $2N+1$ apart and each delay is at most $N-1$. On a fixed tested
phase, job $i$ cannot start before its rounded release date and therefore
incurs tardiness at least equal to its rounding delay. Scheduling every job at
its rounded release date is feasible, so the optimum on each tested grid is
the sum of the delays. Hence $F_{\mathrm{alg}}=N(N-1)/2$, and
\[
 \frac{\Phi_{\mathrm{alg}}}{\Phi^*}
 =\frac{N^2+N(N-1)/2}{N^2}=\frac32-\frac1{2N}.
\]
\end{proof}
This proves tightness for the specified phase-grid output, not approximation
hardness for the shifted problem. Left justifying the algorithm's final job
order in unrestricted continuous time can only improve its value, but it
eliminates the loss on this particular family.

For any prescribed shift $B(I)>0$, the same additive analysis gives
\begin{equation}
 \frac{F_{\mathrm{alg}}+B(I)}{F^*+B(I)}
 \le1+\frac{D(I)}{F^*+B(I)}
 \le1+\frac{pW}{2B(I)}.
 \label{eq:generalshift}
\end{equation}
Thus $B(I)=\lambda pW$ for fixed $\lambda>0$ gives ratio at most
$1+1/(2\lambda)$. Merely increasing the shift can improve this displayed
ratio without improving the schedule, so such a change is not a PTAS for a
fixed objective.

\section{Conclusion}\label{sec:conclusion}
We have shown that the single-machine total weighted tardiness problem with
release dates remains strongly NP-hard when every job has the same processing
time and is individually capable of meeting its due date
(Theorem~\ref{thm:main}). This settles the complexity question documented by
\citet{akker2010} and \citet{gafarov2020}. Complementing the hardness result,
we developed a deterministic phase-grid assignment algorithm that solves at most $N$
minimum-cost assignment problems and runs in $O(N^5)$ arithmetic operations.
For the shifted objective $\Phi=F+p\sum_jw_j$, its approximation ratio is at
most $3/2-1/(2N)$, and this bound is tight for the specified algorithm
(Theorem~\ref{thm:shifted-approx} and Proposition~\ref{prop:grid-tight}).

The normalization argument (Proposition~\ref{prop:normalization}) is the main
structural step in the hardness proof. It applies to every sufficiently
inexpensive feasible schedule and produces one binary choice per vertex, after
which an exact pairwise cost identity
(Proposition~\ref{prop:coefficients} and \eqref{eq:cutcost}) gives the MAX-CUT
correspondence. All job counts and numerical values remain polynomially
bounded (Table~\ref{tab:magnitudes}). For the algorithmic result, the key step
is instead to average the rounding losses over the release phases with
weight-dependent probabilities. This yields the sharper additive error
\eqref{eq:strongadd}, from which the shifted ratio follows. Because an
objective shift changes approximation ratios, this result is not a
multiplicative guarantee for ordinary total weighted tardiness.

The established polynomial cases in the literature remain useful directions
for algorithm design. Further restrictions on release dates, weights, their
ordering, or the number of release phases may support stronger algorithms even
though the general problem is strongly NP-hard. Determining whether the
unshifted problem admits a meaningful multiplicative guarantee on restricted
instance classes, and whether other polynomial algorithms improve on the
phase-grid assignment bound for the shifted objective, remain questions for future work.

\section*{Data and code availability}
% A8.4.6 / B2: correct availability statement to the supplied package
No external empirical data are used. The proofs of
Theorems~\ref{thm:main} and~\ref{thm:shifted-approx} are self-contained and do
not depend on computational experiments. The present manuscript package
contains the \LaTeX{} source with an internal bibliography and its compiled
PDF; it does not include supplementary implementation code or recorded
computational outputs.

\section*{Acknowledgments}
Generative artificial-intelligence tools assisted with proof development, computational checks, literature searches, and manuscript drafting.

% No forced page break: References follow Acknowledgments.
% Superseded external bibliography commands retained, but not executed:
% {\color{blue}\bibliographystyle{plainnat}}
% {\color{blue}\bibliography{references}}
% Internal bibliography: black text; author-name order follows the requested example.
% The author-year citation convention and all 20 original entries are retained.
% The additional 2024 entry is kuehn2024.
{\color{black}

\par}
\end{document}